\documentclass[12pt,draftcls,journal,letterpaper,twosides,onecolumn]{IEEEtran}

\usepackage{amsmath,amssymb,amscd,latexsym,dsfont,wasysym,bbm}
\usepackage{float}
\usepackage{color}
\usepackage{graphicx,subfigure}
\usepackage{multirow,multicol}
\usepackage{verbatim}
\usepackage{psfrag}
\usepackage{comment}
\usepackage{makeidx}
\usepackage[]{algorithm}
\usepackage{algorithmic}
\usepackage{cite}
\usepackage{balance}
\newcommand{\tr}[1]{\textrm{#1}}
\newcommand{\mr}[1]{\mathrm{#1}}
\newcommand{\tnr}[1]{{\textnormal{#1}}}

\newcommand{\mc}[1]{\mathcal{#1}}
\newcommand{\mf}[1]{\mathsf{#1}}

\newcommand{\ms}[1]{\mathds{#1}}

\newcommand{\ov}[1]{\overline{#1}}

\newcommand{\bb}{\boldsymbol{b}}

\newcommand{\bs}{\boldsymbol{s}}

\newcommand{\bx}{\boldsymbol{x}}

\newcommand{\by}{\boldsymbol{y}}

\newcommand{\bz}{\boldsymbol{z}}

\newcommand{\figref}[1]{Fig.~\ref{#1}}

\newcommand{\secref}[1]{Sec.~\ref{#1}}

\newcommand{\exref}[1]{Example~\ref{#1}}

\newcommand{\propref}[1]{Proposition~\ref{#1}}
\newcommand{\corref}[1]{Corollary~\ref{#1}}

\newcommand{\ie}{i.e.,~} 		%note the comma and the space (~)
\newcommand{\eg}{e.g.,~}	%note the comma and the space (~)
\newcommand{\cf}{cf.~}		%note the space (~)

\renewcommand{\emptyset}{\varnothing} % Empty set \varnothing is nicer than \emptyset :)
\newcommand{\argmax}{\mathop{\mr{argmax}}}

\newcommand{\set}[1]{\{#1\}}
\newcommand{\SET}[1]{\left\{#1\right\}}

\newcommand{\cd}{\cdot}
\newcommand{\ld}{\ldots}

\newcommand{\PR}[1]{\Pr\SET{#1}}       	% Probability
\newcommand{\pdf}{p}            			% PDF. I vote for having the PDFs and the PMFs as italic because the pmf as \tr{P}_\bX(\bX) looks weird and also because I like to see the pmfs and the pdfs as functions, which we denote using italic letters (like the Q-function, f(t), s(t), h(t), etc.)

\newcommand{\IND}[1]{\ms{I}_{[#1]}}   	% Indicator function
\newcommand{\Ex}{\ms{E}}     			% Expectation (AA).
\newcommand{\dd}{\,\mr{d}}             		% LS: differentiation operator (\, added to make it nicer)
\newcommand{\mcD}{\mc{D}}

\newcommand{\mcL}{\mc{L}}

\newcommand{\mcR}{\mc{R}}

\newcommand{\mfR}{\mf{R}}
\newcommand{\mfT}{\mf{T}}

\newcommand{\SNR}{\mathsf{snr}}  %%% SNR
\newcommand{\SNRrv}{\mathsf{SNR}}  %%% SNR random variable
\newcommand{\SNRav}{\ov{\mathsf{snr}}}  %%% SNR
\newcommand{\R}{R}              	% spectral efficiency
\newcommand{\Nb}{{\mathop{N_\tnr{b}}}} 	% number of bits in the block
\newcommand{\Ns}{{\mathop{N_\tnr{s}}}} 	% number of symbols in the block
\newcommand{\Nsl}[1]{{\mathop{N_{\tnr{s},#1}}}} 	% number of symbols in the block
\newcommand{\dBval}{\tnr{dB}} % dB (numerical value)
\newcommand{\amc}{\tnr{amc}}
\newcommand{\harq}{{\tnr{harq}}}
\newcommand{\trp}{{\tnr{2r}}}

\newcommand{\nack}{\mf{NACK}}  %% message NACK
\newcommand{\err}{\mf{Err}}  %% decoding error
\newcommand{\kmax}{K}		%% maximum number of transmissions
\newcommand{\PER}{\mathsf{PER}}  % word-error probability (coded)
\usepackage{tikz}
\usetikzlibrary{arrows,shapes,trees}

\definecolor{refkey}{rgb}{0.2,0.8,0.6} % change the colors
\definecolor{labelkey}{rgb}{0.2,0.3,0.4} 
\usepackage[acronym,nonumberlist]{glossaries} % make a separate list of acronyms
\usepackage{hyperref}
\usepackage{amsthm}
 \usepackage[usenames,dvipsnames]{pstricks}
 \usepackage{epsfig}
 \usepackage{pst-grad} % For gradients
 \usepackage{pst-plot} % For axes
\newacronym[\glsshortpluralkey=PDFs,\glslongpluralkey=probability density functions]{pdf}{PDF}{probability density function}
\newacronym[\glsshortpluralkey=CDFs,\glslongpluralkey=cumulative density functions]{cdf}{CDF}{cumulative density function}
\newacronym[\glsshortpluralkey=CCDFs,\glslongpluralkey=complementary cumulative density functions]{ccdf}{CDF}{complementary cumulative density function}
\newacronym[\glsshortpluralkey=PMFs,\glslongpluralkey=probability mass functions]{pmf}{PMF}{probability mass function}
\newacronym[]{lhs}{l.h.s.}{left-hand side}
\newacronym[]{rhs}{r.h.s.}{right-hand side} 

\newacronym[]{bicm}{BICM}{bit-interleaved coded modulation}
\newacronym[]{bicmid}{BICM-ID}{BICM with iterative demapping}
\newacronym[]{cm}{CM}{coded modulation}
\newacronym[]{tcm}{TCM}{trellis-coded modulation}
\newacronym[]{mlc}{MLC}{multi-level coding}
\newacronym[]{pam}{PAM}{pulse amplitude modulation}
\newacronym[]{bpsk}{BPSK}{binary phase shift keying}
\newacronym[]{qam}{QAM}{quadrature amplitude modulation}
\newacronym[]{psk}{PSK}{phase shift keying}
\newacronym[\glsshortpluralkey=LLRs,\glslongpluralkey=logarithmic likelihood ratios]{llr}{LLR}{logarithmic likelihood ratio}
\newacronym[]{map}{MAP}{maximum a posteriori}
\newacronym[]{ml}{ML}{maximum likelihood}
\newacronym[\glsshortpluralkey=MIs,\glslongpluralkey=mutual informations]{mi}{MI}{mutual information}
\newacronym[\glsshortpluralkey=GMIs,\glslongpluralkey=generalized mutual informations]{gmi}{GMI}{generalized mutual information}
\newacronym[]{bicm-gmi}{BICM-GMI}{BICM generalized mutual information}
\newacronym[]{awgn}{AWGN}{additive white Gaussian noise}
\newacronym[]{amc}{AMC}{adaptive modulation and coding}
\newacronym[]{csi}{CSI}{channel state information}
\newacronym[]{cqi}{CQI}{channel quality indicator}
\newacronym[]{sp}{SP}{set-partitioning}
\newacronym[]{gsm}{GSM}{global system for mobile communications}
\newacronym[]{edge}{EDGE}{enhanced data rates for GSM evolution}
\newacronym[]{3gpp}{3GPP}{3rd generation partnership project}
\newacronym[]{lte}{LTE}{Long Term Evolution}
\newacronym[]{dvb}{DVB}{digital video broadcasting}

\newacronym[\glsshortpluralkey=CCs,\glslongpluralkey=convolutional codes]{cc}{CC}{convolutional code}
\newacronym[\glsshortpluralkey=PCCCs,\glslongpluralkey=parallel concatenated convolutional codes]{pccc}{PCCC}{parallel concatenated convolutional code}
\newacronym[\glsshortpluralkey=TCs,\glslongpluralkey=turbo codes]{tc}{TC}{turbo code}
\newacronym{ldpc}{LDPC}{low-density parity-check}
\newacronym[]{ofdm}{OFDM}{orthogonal frequency-division multiplexing}
\newacronym[]{bep}{BEP}{bit-error probability}
\newacronym[]{wep}{WEP}{word-error probability}
\newacronym[]{sep}{SEP}{symbol-error probability}
\newacronym[]{ttcm}{TTCM}{turbo-trellis coded modulation}
\newacronym[]{uep}{UEP}{unequal error protection}
\newacronym[\glsshortpluralkey=CENCs,\glslongpluralkey=convolutional encoders]{cenc}{CENC}{convolutional encoder}
\newacronym[]{mimo}{MIMO}{multiple-input multiple-output}
\newacronym[\glsshortpluralkey=SNRs,\glslongpluralkey=signal-to-noise ratios]{snr}{SNR}{signal-to-noise ratio}
\newacronym[]{msb}{MSB}{most-significative bit}
\newacronym[]{bcjr}{BCJR}{Bahl--Cocke--Jelinek--Raviv}
\newacronym[\glsshortpluralkey=SEDs,\glslongpluralkey=squared Euclidean distances]{sed}{SED}{squared Euclidean distance}
\newacronym[\glsshortpluralkey=EDs,\glslongpluralkey=Euclidean distances]{ed}{ED}{Euclidean distance}
\newacronym[\glsshortpluralkey=MEDs,\glslongpluralkey=minimum Euclidean distances]{med}{MED}{minimum Euclidean distance}
\newacronym[]{core}{CoRe}{constellation rearrangement}
\newacronym[]{msd}{MSD}{multistage decoding}
\newacronym[]{pdl}{PDL}{parallel decoding of the individual levels}
\newacronym[\glsshortpluralkey=GCs,\glslongpluralkey=Gray codes]{gc}{GC}{Gray code}
\newacronym[]{brgc}{BRGC}{binary-reflected Gray code}
\newacronym[]{nbc}{NBC}{natural binary code}
\newacronym[]{fbc}{FBC}{folded-binary code}
\newacronym[]{bsgc}{BSGC}{binary semi-Gray code}
\newacronym[]{msp}{MSP}{modified set-partitioning}
\newacronym[]{ssp}{SSP}{semi set-partitioning}
\newacronym[]{fhd}{FHD}{free Hamming distance}
\newacronym[]{mfhd}{MFHD}{maximum free Hamming distance}
\newacronym[]{ods}{ODS}{optimal distance spectrum}
\newacronym[]{iud}{i.u.d.}{independent and uniformly distributed}
\newacronym[]{ud}{u.d.}{uniformly distributed}
\newacronym[]{iid}{i.i.d.}{independent, identically distributed}
\newacronym[]{bico}{BICO}{binary-input continuous-output}
\newacronym[]{gh}{GH}{Gauss--Hermite}
\newacronym[\glsshortpluralkey=BSs,\glslongpluralkey=base-stations]{bs}{BS}{base-station}
\newacronym[\glsshortpluralkey=MSs,\glslongpluralkey=mobile-stations]{ms}{MS}{mobile-station}

\newacronym[]{phy}{PHY}{physical layer} 
\newacronym[]{llc}{LLC}{logical link control} 
\newacronym[]{mac}{MAC}{media access control} 
\newacronym[]{fft}{FFT}{fast Fourier transform} 
\newacronym[]{cf}{CF}{characteristic function} 
\newacronym[]{mgf}{MGF}{moment generating function} 
\newacronym[]{ee}{EE}{energy efficiency} 
\newacronym[]{kkt}{KKT}{Karush--Kuhn--Tucker} 
\newacronym[\glsshortpluralkey=MCSs,\glslongpluralkey=modulation/coding schemes]{mcs}{MCS}{modulation/coding scheme} 
\newacronym[]{fec}{FEC}{forward error correction}
\newacronym[]{arq}{ARQ}{automatic repeat request}
\newacronym[]{harq}{HARQ}{hybrid ARQ}
\newacronym[]{trp}{2R-HARQ}{two-rounds HARQ}
\newacronym[]{tarq}{TARQ}{truncated HARQ}
\newacronym[]{ccharq}{HARQ-RR}{repetition redundancy HARQ}
\newacronym[]{rrharq}{HARQ-RR}{repetition redundancy HARQ}
\newacronym[]{irharq}{HARQ-IR}{incremental redundancy HARQ}
\newacronym[]{vlharq}{HARQ-VL}{variable-length HARQ}
\newacronym[]{pdharq}{HARQ-PD}{packet-dropping HARQ}
\newacronym[]{ack}{ACK}{positive acknowledgment}
\newacronym[]{nack}{NACK}{negative acknowledgment}
\newacronym[]{crc}{CRC}{cyclic redundancy check}
\newacronym[]{dp}{DP}{dynamic programming}
\newacronym[]{gp}{GP}{geometric programming}
\newacronym[]{per}{PER}{packet error rate}
\newacronym[]{op}{OP}{outage probability}
\newacronym[]{spa}{SPA}{saddle-point approximation}
\newacronym[]{mrc}{MRC}{maximum ratio combining}
\newacronym[]{mdp}{MDP}{Markov decision process}
\newacronym[]{pf}{PF}{Proportional Fairness}
\newacronym[]{rr}{RR}{Round Robin}

\newtheorem{corollary}{Corollary}
\newtheorem{proposition}{Proposition}

\newtheorem{example}{Example}
\newcommand{\siz}{1.2}  %%%%%taille des figures 0.85 two columns, 1.2 one column
\newcommand{\sizf}{0.6}   %%%%%taille des figures 0.9 two columns, 0.65 one column

\title{HARQ and AMC: Friends or Foes?}

\author{Redouane Sassioui, Mohammed Jabi, Leszek Szczecinski, Long Bao Le, Mustapha Benjillali, and Benoit Pelletier
\\
\thanks{This work was submitted in part to IEEE Global Communications Conference 2016, Washington, DC, USA.
}
\thanks{%
R.~Sassioui, M.~Jabi, L.~Szczecinski, and L.~B.~Le are with INRS-EMT, Montreal, Canada. [e-mail: \{sassioui,~jabi,~leszek,~long.le\}@emt.inrs.ca]}%
\thanks{%
M. Benjillali is with the Communication Systems Department, INPT, Rabat, Morocco. [e-mail: benjillali@ieee.org]}%
\thanks{%
Benoit Pelletier is with InterDigital Communications, Montreal, Canada. [e-mail: Benoit.Pelletier@interdigital.com]}%
\thanks{%
The work was supported by the government of Quebec, under grant \#182388.}%
}%

\begin{document}

\maketitle

\begin{abstract}
To ensure reliable communication in randomly varying and error-prone channels, wireless systems use \gls{amc} as well as \gls{harq}. In order to elucidate their compatibility and interaction, we compare the throughput provided by \gls{amc}, \gls{harq}, and their combination (\gls{amc}-\gls{harq}) under two operational conditions: in slow- and  fast block-fading channels. Considering both,  \gls{irharq} and \gls{rrharq} we optimize the rate-decision regions for \gls{amc}/\gls{harq} and compare them in terms of attainable throughput. Under a fairly general model of the channel variation and the decoding functions, we conclude that i)~adding \gls{harq} on top of \gls{amc} may be counterproductive in the high average signal-to-noise ratio regime for fast fading channels, and ii)~\gls{harq} is useful for slow fading channels, but it provides moderate throughput gains. We provide explanations for these results which allow us to propose  paths to improve AMC-HARQ systems.
\end{abstract}


%!TEX root =  ../HARQ-AMC.tex
\section{Introduction}

\gls{amc} and \gls{harq} are two transmission strategies commonly used in modern wireless systems to communicate over error-prone and time-varying channels. The main objective of this work is to compare \gls{amc} and \gls{harq} from the throughput point of view, which is a well-established reference criterion particularly suited for transmission of delay-insensitive contents.

In broad terms, \gls{amc} consists in adjusting the transmission parameters (such as the modulation type, the coding rate, and/or the transmission power) to the channel conditions \cite{Alouini00,Goeckel99}; these are most often defined by a predefined set of \glspl{mcs}. The receiver selects a suitable \gls{mcs} and conveys its index to the transmitter via a feedback channel. 

Transmission errors, unavoidable in any practical system, are  handled by the retransmission protocol known as \gls{arq}\cite{Lin03_Book}, where the receiver uses a feedback channel to inform the transmitter about a successful decoding---via a \gls{ack} message---or about a decoding failure---via a \gls{nack} message. Each \gls{nack}  triggers a new \emph{transmission round} (or a \emph{retransmission}), and increasing their number improves reliability. In this work, we consider the retransmission protocol known as \emph{hybrid} \gls{arq} (\gls{harq}) in which coding is done across the transmission rounds and thus  is intimately related to the \gls{amc} for which coding and modulation are the core elements.

In more general terms, both \gls{amc} and \gls{harq} may be seen as variable rate transmission strategies: the former, due to the explicit variation of the number of bits conveyed over the channel, the latter, due to the variability of the transmission time resulting from variable number of transmission rounds. Both are thus naturally coupled and, in this work we want to clarify to what extent this coupling should be preserved or exploited.

Therefore, although the formalism of the communication layers tend to separate the \gls{amc} from \gls{harq}, the practice may call for their holistic view.  In particular, the  \gls{lte} standard specifies the \gls{harq} operation as part of the \gls{mac} layer \cite[Sec.~5.4.2]{3GPP_TS_36.321}, whereas the channel measurements procedures are defined in the \gls{phy} \cite[Sec.~7.2]{3GPP_TS_36.213}. However,  the \gls{amc}, \ie the way the \gls{mcs} should be chosen based on the reported channel measurement is  unspecified and left for implementation. We will consider the \gls{amc} and \gls{harq} as mechanisms of \gls{phy}.

The main difficulty in providing a qualitative insight into the relationship between \gls{harq} and the \gls{amc} lies in the fact that they are affected by various elements, such as the adopted channel model, the available coding/modulation schemes, or the power adaptation strategies. It is thus critical to strike a realistic balance between the general model and the plethora of clumsy technical details of the practical systems or standards. With this idea in mind, the channel model must be kept simple and we will consider two extreme cases of the operating conditions with respect to the channel correlation between \gls{harq} rounds: i)~in fast block-fading channels, all transmission rounds are carried out over independent fading realizations, and  ii)~in slow block-fading channels, all \gls{harq} transmission rounds (of the same packet) are carried out in the same operating conditions (are thus perfectly correlated). Under these models, the relationship between \gls{amc} and \gls{harq} becomes tractable while still allowing us to extrapolate the findings to the situations where the channel correlation is only partial -- a case which is difficult to tackle analytically and a numerical approach might be necessary.

Another important point is the choice of a suitable performance criterion. The throughput seen by the upper layers provides a simple and fair comparison baseline, but ignores most of the considerations related to delay or packet loss. The former is acceptable by  assuming a delay-insensitive traffic and a saturated buffer operating mode; the latter requires a few assumptions about the upper \gls{llc} layer taking care of all residual retransmission errors which then become irrelevant to the analysis (more on this in \secref{Sec:LLC}).

The numerical examples that will illustrate our conclusions and findings will be based on a simple model of the decoding errors which may be fit to the experimental data obtained with practical decoders. We will then use simplifications to characterize the behavior of \gls{irharq} when different parts of the codeword are transmitted in the different rounds \cite{Hagenauer88,Caire01}, and of \gls{rrharq} when each round carried the same codeword \cite{Benelli85,Chase85}.

%contribution
The contributions of this work can be summarized as follows:
\begin{itemize}
\item We provide a general model which leads to a qualitative evaluation of the relationship between \gls{amc} and \gls{harq}. We provide the throughput expressions valid for fast- and slow-fading channel models, and compare to the formulations appearing in the literature.
\item We discuss how the transmission rate should be chosen based on the observed \gls{csi}, to maximize the throughput of \gls{amc} and \gls{amc}-\gls{harq} for fast- and slow-fading channels. 
\item We prove that, in the high average \gls{snr} regime, combining \gls{harq} with the \gls{amc} affects negatively the system throughput in fast-fading channel, and we provide an intuitive explanation of this result. We  also show that \gls{harq} is always productive for slow-fading channel.
\item Finally, we propose en evaluate two strategies to enhance the performance of \gls{harq} for fast-fading channels which remove the throughput penalty introduced by the conventional \gls{harq}.
\end{itemize}

The rest of this paper is organized as follows. The system model is described in \secref{Sec:Model}, where -- somewhat unconventionally -- \secref{Sec:state.of.art}  positions this work with respect to the literature. We found it the most appropriate because \gls{amc} and \gls{harq} have been studied under many different modeling assumptions, and it is  difficult to discuss them before the model is introduced.

The throughput analysis of  \gls{amc} and \gls{harq} is presented in \secref{Sec:AMC}, \secref{Sec:slow.fading}, and \secref{Sec:Block.fading}. The directions towards the improvement of AMC-HARQ systems are explored in \secref{Sec:Boosting.HARQ} and we present our conclusions in \secref{Sec:Conclusions}.
%!TEX root =  ../HARQ-AMC.tex
\section{Model}\label{Sec:Model}

%%%%%%%%%%%%%%%%%%%%%%
\subsection{Channel Model}
We consider the communication over a block flat-fading channel, so the received signal can be written as
\begin{align}\label{y.s.snr}
\by[n]=\sqrt{\SNR[n]}\bs[n]+\bz[n],
\end{align}
where $\bz[n]$ is a zero-mean, unit-variance complex Gaussian noise, $\bs[n]$ is the signal composed of $\Ns$ symbols carrying the encoded message, and $\SNR[n]$ is the value of the \gls{snr} in the block-time $n$.

We assume that the user is granted access to the channel in predefined, non-consecutive, time blocks. The time-difference between the beginning of the  blocks $\bs[n], \bs[n+1],\ld$ is given by $(1+q)\Ns$. It is convenient to think about $q$ as an integer, which should be interpreted as the number of blocks left between two consecutive transmission.

We assume also a  transmission with a constant power and, thus, the \gls{snr} $\SNR[n]$ captures the variation of the channel. The relationship between the \gls{snr}s $\SNR[n], \SNR[n+1]\ld$, depends on the channel coherence time $\tau_\tr{coh}$, and we consider two different cases%
\begin{enumerate}
\item \textbf{Fast fading}, where $q\Ns\gg\tau_\tr{coh}$, and then the \gls{snr}s $\SNR[n], \SNR[n+1], \ld$ may be modeled as independent random variables. This corresponds to the case of fast-moving users where channel conditions change quickly from one block to another. 
\item \textbf{Slow fading}, where $q\Ns \ll \tau_\tr{coh}$, and then the transmissions experience the same \gls{snr} in many subsequent blocks, \ie $\SNR[n]\approx\SNR[n+1]\approx\SNR[n+2]\approx\ld$; this model is suitable for slowly-moving users.
\end{enumerate}

In the numerical examples, we will use the Rayleigh fading so the \gls{snr} is modeled by a random variable $\SNRrv$ following an exponential distribution
\begin{align}\label{PDF_Nakagami_inst_SNR}
\pdf_{\SNRrv}(\SNR)=(1/\SNRav)\exp\bigl(-\SNR/\SNRav\bigr),
\end{align}
where $\SNRav$ is the average \gls{snr}.%

%%%%%%%%%%%%%%%%%%%%%%%%%%%%%%
\subsection{Physical Layer (PHY): AMC and HARQ}
 
 %%%%%
The role of \gls{amc} consists in encoding the information bits, $\bb$ obtained from the \gls{llc}, and transmitting them over the channel. We assume that the transmission rates to be adopted are taken from the set $\mcR=\set{R_l}_{l=1}^L$ (measured in [bits/symbols]); each transmission thus carries an encoded version of $R_lN_s$ bits. This corresponds to  the reality of current systems which use predefined sets of transmission rates, each corresponding to a particular \gls{mcs} supported by both the transmitter and the receiver \cite[Sec.~7.2]{3GPP_TS_36.213}. This is  different from the approach used in \cite{Kim07}, where the optimized rates depended on $\SNRav$. We note that the performance of the system is, of course, affected by the choice of \gls{mcs}s, nevertheless, the methodology of comparison and the main conclusions we draw are valid independently of a particular choice.

The transmitter uses the rate $R_{\hat{l}}$, where $\hat{l}\in\set{1,\ld,L}$ is the \gls{mcs} index sent by the receiver over the error-free feedback channel, see \figref{fig:HARQ.AMC.model}. Due to the block-fading model \eqref{y.s.snr}, it is enough to discretize the \gls{snr}, $\SNR$ which we also assume to be perfectly estimated at the receiver, \ie
\begin{align}\label{hat.l}
\hat{l}=\sum_{l=1}^L  l~\IND{\SNR\in \mcD_l},
\end{align} 
where $\mcD_l$ is the \gls{mcs} \emph{decision region} that can be adjusted in order to maximize the criterion of interest (here, the throughput). It will be formally defined later.

%$\gamma_1\leq\gamma_2\leq\ld\leq\gamma_L$ are the \gls{snr} discretization intervals and we set $\gamma_1\triangleq 0$ and $\gamma_{L+1}\triangleq \infty$; here, $\IND{a}=1$ if $a$ is true, and $\IND{a}=0$, otherwise. We will assume that \gls{cqi} is sent (over the feedback channel) without delay or errors to the transmitter.

%Since the set of rates, $\mcR$, is fixed, we can adjust $\set{\gamma_l}_{l=2}^{L}$ to optimize the performance, as we will do later.

%%%%%%%
\begin{figure}[tb]
\hspace{-0.5cm}
\scalebox{.7}{\pgfdeclarelayer{background}
\pgfdeclarelayer{foreground}
\pgfsetlayers{background,main,foreground}

% Define a few styles and constants

\tikzstyle{form1} = [draw, minimum width=2cm, text width=1.8cm, fill=blue!15, 
  text centered,  minimum height=.9cm]

\tikzstyle{form2} = [draw=none, minimum width=3.4cm, text width=1.5cm, fill=none, 
  text centered,  minimum height=0cm]    

\tikzstyle{form3} = [draw,dashed, minimum width=1.8cm, text width=1.8cm, fill=red!10, 
  text centered,  minimum height=1.8cm]
  
  \tikzstyle{form4} = [draw=none, minimum width=0cm, text width=0cm, fill=none, 
  text centered,  minimum height=.2cm] 
  \tikzstyle{form5} = [draw=none, minimum width=0cm, text width=0cm, fill=black, 
  text centered,  minimum height=.6cm] 
    
    \tikzstyle{elip}=[draw,ellipse, minimum width=.2cm, text width=0cm, fill=none, 
  text centered,  minimum height=1cm]
  
    \tikzstyle{elip2}=[draw,ellipse, minimum width=.1cm, text width=0cm, fill=none, 
  text centered,  minimum height=.6cm]

\begin{tikzpicture}

  %%%%%%%%%% noeuds  %%%%%%%%%%%%%%%%  
    \node(Channel)[form1] {Channel};
    
     \path (Channel.west)+(-2.7,0)  node[form1](Enc) {Encoder};
     \path (Channel.east)+(2.7,0)  node[form1](Dec) {Decoder};
     %\path (Channel.north)+(0,1.6)  node[form3](FedBack) {Feedback Channel};
       \path (Enc.north)+(0,1.6)  node[form1](AMC_cont1) {AMC};
        \path (AMC_cont1.north)+(0,1.6)  node[form1](HARQ_cont1) {HARQ Controller};
         \path (HARQ_cont1.north)+(0,1.6)  node[form1](ARQ_cont1) {ARQ Controller};
          \path (HARQ_cont1.west)+(-0.8,1)  node(P1) {};
      \path (Dec.north)+(0,3.657)  node[form1](HARQ_cont2) {HARQ Controller};
       %\path (AMC_cont2.north)+(0,1.6)  node[form1](HARQ_cont2) {HARQ Controller};
         \path (HARQ_cont2.north)+(0,1.6)  node[form1](ARQ_cont2) {ARQ Controller};
       \path (HARQ_cont2.east)+(0.8,1)  node(P2) {};
       \path (Enc.west)+(-1,-0.2)  node(Buff1) {};
       \path (Enc.south)+(0,-1)  node(Tx) {Transmitter};
        \path (Dec.south)+(0,-1)  node(Rx) {Receiver};
        
         \path(ARQ_cont1.north) +(0,0.6) node(base1){};
       \path(ARQ_cont1.north) +(1,0.6) node(Buff11){};
       \path(ARQ_cont1.north) +(-1,0.6) node(Buff12){};
       \path(ARQ_cont1.north) +(1,0.8) node(Buff13){};
       \path(ARQ_cont1.north) +(-1,0.8) node(Buff14){};
       \path(ARQ_cont1.north) +(1,1) node(Buff15){};
       \path(ARQ_cont1.north) +(-1,1) node(Buff16){};
       \path(ARQ_cont1.north) +(1,1.6) node(Buff17){};
       \path(ARQ_cont1.north) +(-1,1.6) node(Buff18){};

       \path(ARQ_cont2.north) +(0,0.6) node(base2){};
       \path(ARQ_cont2.north) +(1,0.6) node(Buff21){};
       \path(ARQ_cont2.north) +(-1,0.6) node(Buff22){};
       \path(ARQ_cont2.north) +(1,0.8) node(Buff23){};
       \path(ARQ_cont2.north) +(-1,0.8) node(Buff24){};
       \path(ARQ_cont2.north) +(1,1) node(Buff25){};
       \path(ARQ_cont2.north) +(-1,1) node(Buff26){};
       \path(ARQ_cont2.north) +(1,1.6) node(Buff27){};
       \path(ARQ_cont2.north) +(-1,1.6) node(Buff28){};
        %%%%%%%%%%%%%%%%%%%%%%%%%%SNRs%%%%%%%%%%%
%        \path (Enc.south)+(-2,-2)  node[form5](SNR_0) {};
%        \path (SNR_0.east)+(2,0)  node[form5](SNR_1) {};
%        \path (SNR_1.east)+(3,0)  node[form5](SNR_2) {};
%        \path (SNR_2.east)+(1.5,0)  node[form5](SNR_3) {};
%        \path (SNR_3.east)+(2,0)  node(SNR_4) {};
%         \path (SNR_0.north)+(0,0.3)  node(gamma0) {$\gamma_1$};
%         \path (SNR_1.north)+(0,0.3)  node(gamma1) {$\gamma_2$};
%         \path (SNR_2.north)+(0,0.3)  node(gamma2) {$\gamma_3$};
%         \path (SNR_3.north)+(0,0.3)  node(gamma3) {$\gamma_L$};
%         \path (SNR_4.north)+(0,0.3)  node(gamma4) {$\gamma$};
%

 \path[draw] (Buff11.west)--(Buff12.east);
     \path[draw] (Buff13.west)--(Buff14.east);
     \path[draw] (Buff15.west)--(Buff16.east);
     \path[draw] (Buff11.west)--node[shift={(0.5,0)}]{$\bigg\} \bb$}(Buff17.west);
     \path[draw] (Buff12.east)--(Buff18.east);
    \path[draw,<-,>=latex] (ARQ_cont1.north)--(base1.base);

    \path[draw] (Buff21.west)--(Buff22.east);
     \path[draw] (Buff23.west)--(Buff24.east);
     \path[draw] (Buff25.west)--(Buff26.east);
     \path[draw] (Buff21.west)--(Buff27.west);
     \path[draw] (Buff22.east)--(Buff28.east);
    \path[draw,->,>=latex] (ARQ_cont2.north)--(base2.base);

 \path[draw,->,>=latex] (Enc.east)--node[shift={(0,0.3)}]{$\bs[n]$}(Channel.west);
  \path[draw,->,>=latex] (Channel.east)--node[shift={(0,0.3)}]{$\by[n]$}(Dec.west);  
  
  %\path[draw,dashed,<-,>=latex] (Buff1P2.south)+(1,0)--(Buff1P2.south);
%   \path[draw,<-,>=latex] (Buff1.base)+(1,0)--(Buff1.base);
%     \path[draw,->,>=latex] (Buff2.base)+(-1,0)--(Buff2.base);
%     

     \path[draw,->,>=latex] (AMC_cont1.south)--node[shift={(0.5,0)}]{$R \in \mcR$}(Enc.north);
     \path[draw,<-,>=latex] (HARQ_cont2.south)--(Dec.north);
%\path[draw,<-,>=latex] (HARQ_cont2.south)--(AMC_cont2.north);
      \path[draw,->,>=latex] (HARQ_cont1.south)--(AMC_cont1.north);
     \path[draw,<-,>=latex] (ARQ_cont2.south)--(HARQ_cont2.north);
      \path[draw,->,>=latex] (ARQ_cont1.south)--(HARQ_cont1.north);
     
     %\path[draw] (AMC_cont2.west)--node[shift={(0,.3)}]{$\tr{CSI}(t)$}(FedBack.east);
      \path (Dec.north)+(-0.5,0)  node(INT) {};
       \path (INT.north)+(0,1.6)  node(INT2) {};
            \path[draw,dashed]  (INT.north) --(INT2.south) {};

      \path[draw,<-,>=latex,dashed] (AMC_cont1.east)--node[shift={(0,.3)}]{$\tr{MCS}, \hat{l}$}(INT2.south);
      \path[draw,<-,>=latex,dashed] (ARQ_cont1.east)--node[shift={(0,.3)}]{$\tr{(N)ACK}$}(ARQ_cont2.west);
       \path[draw,<-,>=latex,dashed] (HARQ_cont1.east)--node[shift={(0,.3)}]{$\tr{(N)ACK}$}(HARQ_cont2.west);

      \path[draw,>=latex,dashed] (P1.east)--(P2.west);
       \path (P1.north)+(0,0.3)  node(gamma0) {LLC};
        \path (P1.south)+(0,-0.3)  node(gamma0) {PHY};
        \path (P2.north)+(0,0.3)  node(gamma0) {LLC};
        \path (P2.south)+(0,-0.3)  node(gamma0) {PHY};
      %\draw[style=dashed] AMC_cont1.east)--node[shift={(0,.3)}]{$\tr{CSI}(t-\tau)$}(AMC_cont2.west);
       
      % \path[draw] (P1_fed.base)-| node[shift={(-1.2,.2)}]{$I[n]$}(P1_Dec.base);
      %\path[draw] (P1_fed.base)-| node[shift={(-1.2,.2)}]{$I[n]$}(P1_Dec.base);
     % \path[draw](P2_fed.base)-| node[shift={(0.9,.-1)}]{$CSI(t-\tau)$}(P1_Enc.base);
      % \path [draw,->,>=latex] (P1_Enc.base)--(AMC_cont1.east);
        %\path[draw](P1_fed.base)-| node[shift={(-0.6,.-1)}]{$CSI(t)$}(P1_Dec.base);
        %\path [draw] (P1_Dec.base)--(AMC_cont2.west);

     % \path[draw,->,>=latex] (AMC_cont1.west)-|(el.north);
       %\path[draw,->,>=latex] (Arq_cont2.east)-|(el2.north);

 %%%%%%%%%%%%%%%%%%%%%%%%%%%%%%%%%%%%%%%

\end{tikzpicture}}
\caption{Model of the transmission with the \gls{arq} implemented by the \gls{llc}; \gls{amc} and \gls{harq} belong to \gls{phy}. %Here, the decoder also estimates the \gls{snr} and sends back the \gls{mcs} index, $\hat{l}$.
}\label{fig:HARQ.AMC.model}
\end{figure}
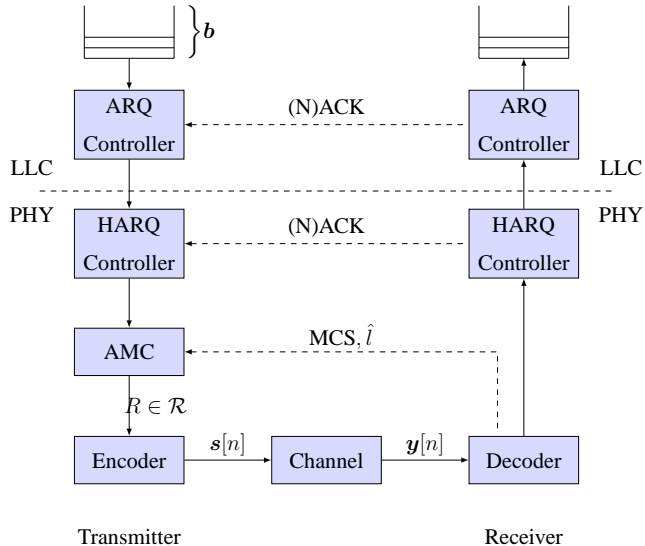

Another role of \gls{amc} is to prepare the upcoming \gls{harq} transmission rounds. To this end, \gls{amc} encodes the packet into $K$ sub-codewords $\bx_{1},\ld, \bx_{K}$ of equal length. The role of \gls{harq} is to transmit successively these sub-codewords upon reception of \gls{nack} messages. For the moment, we assume that each sub-codeword occupies the whole available block $\bs[n]$, thus, the information in each round is transmitted with the  rate $\R_l$ determined by the \gls{amc} in the first round. This is an important assumption, which allows us to focus the analysis; we will relaxe it later.

%We note that \gls{irharq} can use the variable rate transmission \cite{Szczecinski13} which will be discussed in \secref{Sec:vlharq}. 

In \gls{harq}, the information packet $\bb$ is sent using many channel blocks, so it is convenient to use a packet-oriented notation and denote the respective variables using subindices, \eg $\by_{t}$ and $\SNR_t$ will denote the channel outcome and the respective \gls{snr} of the $t$-th round of the packet $\bb$.

The \gls{phy}-related actions terminate when \gls{ack} is received or when the final $K$-th \gls{harq} round is reached. The receiver then discards the channel outcomes and sends a final acknowledgement that is shared with the \gls{llc} layer which takes over the communication process.

%Our objective in this work is to calculate--and then, to compare--the throughout of the three particular systems defined as:
%\begin{enumerate}
%\item ``Pure" AMC, where $K = 1$, which means that retransmissions are not allowed and HARQ is not used.
%\item ``Pure" HARQ, where $L = 1$ and $K >1$, which means that one rate used at PHY and no rate adaptation to channel condition is used. However, K retransmissions by HARQ are allowed.
%\item AMC-HARQ, where $L >1$ and $K >1$, which means that both, AMC and HARQ, are used at PHY.
%\end{enumerate}
%%%%%%%%%%%%%%%%%%%%%%%%%%%%%%
\subsection{Logical Link Control Layer}\label{Sec:LLC}
%In order to describe meaningfully the behaviour of our system, it is necessary to go beyond the description of \gls{phy} and the \gls{llc}. As we have shown, most of the previous work rely on the implicit assumption about the operation of LLC and we thus briefly describe the simple model of the layers we consider in the protocol stack.

In our model, the final acknowledgement of \gls{harq} is shared with the \gls{llc} layer as shown schematically in \figref{fig:HARQ.AMC.model}. \gls{llc} ignores the details of the operations of \gls{phy} and only relies on the final \gls{ack}/\gls{nack}, implementing thus a basic form of \gls{arq}: upon reception of \gls{ack}, the packet $\bb$ is removed from the \gls{llc} buffer; if \gls{nack} is received, the contents (bits) of $\bb$ are kept in the buffer. Then, a new packet $\bb'$ is formed, which contains $\Ns \R_{\hat{l}}$ bits ($\hat{l}$ is the \gls{mcs} index obtained by the transmitter at the moment it has to form the packet $\bb'$). The new packet $\bb'$ may contain some (or all) of the bits from the previously ``NACKed" packet $\bb$ which, in general, used a rate $\R_k\neq\R_{\hat{l}}$. 

The feedback channel carrying one-bit \gls{ack}/\gls{nack} messages at \gls{phy} and \gls{llc} is assumed to be error-free as. theoretically, these bits can be protected with arbitrary strength, whose overhead may be neglected for sufficiently large $\Ns$.\footnote{In practice, the acknowledgement messages for \gls{harq} are  received with a non-zero error probability; \eg the \gls{lte} specifies a minimum \gls{snr} requirement for reception of the \gls{ack} messages with errors no greater than 1\% \cite[Sec.~8.3.2.1]{3GPP_TS_36.104}.  On the other hand, the acknowledgement messages at the \gls{llc} are grouped for various packets and protected against errors through coding; then they are considered error-free. For tractability, we do not include these elements in our model.}

Since the packet is discarded from the \gls{llc} buffer only after \gls{ack} is received, there is no loss of data independently of how unreliable \gls{phy} is. We thus implicitly assumed that we deal with delay-tolerant but loss-sensitive applications, which is justified, for example, when files with critical contents are being transmitted. 

Because it might not be immediately obvious, we emphasize that implementing \gls{arq} at the \gls{llc} does not change the \gls{llc} throughput seen by the upper layers.  In other words, retransmitting the NACKed packets does not degrade the throughput because  this criterion is blind to which bits are actually being transmitted -- the fact that the same bits are retransmitted  is irrelevant in the throughput evaluation.  
In fact, the throughput remains the same independently of the number of allowed \gls{arq} rounds at \gls{llc}, \ie truncated \gls{arq} does not change the throughput neither.\footnote{In practice, the addition delay induced by the retransmissions may be a source of the throughput degradation. For example, the TCP may interpret the delay as a congestion and respond by decreasing the size of the transmitted packets, which lowers the throughput. %Again, we do not take into consideration such a situation which may be avoided if connection-less protocols, such as UDP are employed.
}

This is formally stated in \cite{Jabi16}, and it holds in absence of other system-related considerations such as the communication overhead in the feedback channel, or the buffer overflow probability, which we ignore here. In our case, assuming an unlimited number of \gls{arq} rounds leads to the lossless communication which is particularly useful from the theoretical point of view, as it allows us to make a fair comparison between various \gls{phy} strategies using a single criterion -- the throughput.

%%%%%%%%%%%%%%%%%%%%%%%%%%%%%%%%%%%%%%%%%%%%

%%%%%%%
\subsection{Decoding Errors: AMC}\label{Sec:Errors.AMC}

The probability of error in the first transmission round, $\err_1$, depends on the experienced \gls{snr} and the adopted rate $R_l$, \ie $\PR{\err_1|R_l,\SNR}=\PER_l(\SNR)$ and  it might be established experimentally for a given encoding and decoding. However, for the numerical analysis, it is convenient to use the parametric description of the \gls{per} function
\begin{align}
\label{PER.SNR}
\PER_l(\SNR)&=
\begin{cases}
1 &\text{if}\quad \SNR<\SNR_{\tr{th},l}\\
\exp[-\tilde{a}_l (\SNR/\SNR_{\tr{th},l}-1) ] &\text{if}\quad \SNR\geq\SNR_{\tr{th},l}
\end{cases},%\\
\end{align}
where the thresholds $\SNR_{\tr{th},l}$ and the decay parameters $\tilde{a}_l$ should be found from the empirical/measured data. A form similar to \eqref{PER.SNR} has been also used in \cite{Liu04}, which provided tabulated values of $\exp(\tilde{a}_l)$, $\SNR_{\tr{th},l}$, and $\tilde{a}_l/\SNR_{\tr{th},l}$ for a particular class of encoders/modulators and decoders. 

We found that, a common decay value $\tilde{a}_l=\tilde{a}$ allows for a compact description of the \gls{snr}-\gls{per} relationship for the same family of encoders/decoders: for convolutional codes, $\tilde{a}\approx 4$ fits well the experimental data, while $\tilde{a}\approx 15$ should be used for turbo-coded transmissions with large codewords.\footnote{\cite[Table I]{Liu04} shows the values $a_l=\exp(\tilde{a}_l)$, \ie $\tilde{a}_l=\log{a}_l\in\set{5.6, 4.4, 4.2, 3.9, 4.0, 3.5}$. It is not clear how these $a_l$ were obtained, but the \gls{per} curves in \cite[Fig.~10]{Liu04} may be approximated using a common value $\tilde{a}_l=4$.} Here, instead of matching the value of $\tilde{a}$ to a particular family of \gls{per} curves as done in \cite{Liu04}, we treat $\tilde{a}$ as a parameter which defines how quickly the \gls{per} curve decays as a function of \gls{snr}. Its operational meaning may be defined from \eqref{PER.SNR} in terms of \gls{snr}, $\Delta\cd\SNR_{\tr{th},l}$, necessary to achieve the desired level of the \gls{per}, $\PER_l(\Delta\cd\SNR_{\tr{th},l})=\epsilon$, as
\begin{align}\label{Delta.epsilon}
\Delta={\ln(1/\epsilon)}/{\tilde{a}}+1.
\end{align}
For example, setting $\epsilon=10^{-2}$, we obtain %$\Delta=1.1~\!\tr{dB}$ for $\tilde{a}=15$, 
$\Delta=3.3\dBval$ for $\tilde{a}=4$, and $\Delta=10\dBval$ for $\tilde{a}=0.5$.

We also use $I(\SNR_{\tr{th},l})=\R_l$, where $I(\SNR)=\mf{I}(S;Y)$ is the mutual information between the random variables $S$ and $Y$ modeling the transmitted and the received signals $\bs[n]$ and $\by[n]$. This coincides with the threshold decoding of the capacity achieving codes which can be modelled with $\tilde{a}=\infty$; it is convenient because, the results are then comparable to those already presented in the literature for \gls{amc} \cite{Kim07} or \gls{harq} \cite{Caire01}.

\subsection{Decoding Errors: HARQ}\label{Sec:Errors}

Let $\nack_k=\set{\err_1,\ld, \err_k}$ be the event of $k$ consecutive decoding errors. Each error event $\err_t$, $t=1,\ld, k$, depends on the \gls{snr}s experienced in $t$ rounds and on the transmission rate $\R=\R_l$ adopted in the first round. As in \secref{Sec:Errors.AMC}, the probability of this event is given by the \gls{per} function $\PR{\err_t|\R_l}=\PER_l(\SNR_1, \ld,\SNR_t)$ which now depends on $t$ different \gls{snr}s. However, a multidimensional representation of the \gls{per} is not tractable, and we adopt the simplifying approach of \cite{Wan06, Pauli2007} which reduces the multidimensional function to the following  scalar representation
\begin{align}\label{MD.2.1D}
\PER_l(\SNR_1,\ld,\SNR_k)\approx\PER_l(\SNR^\Sigma_k),
\end{align}
where  the ``aggregate'' \gls{snr} obtained as
\begin{align}\label{aggreg.SNR}
\SNR^\Sigma_k=h^{-1}\Big( \sum_{t=1}^k h(\SNR_t)\Big);
\end{align}
the function $h(\cd)$ and its inverse $h^{-1}(\cd)$ depend on the encoding; \ie on the type of \gls{harq}. 

For \gls{rrharq}, each codeword is the same, \ie $\bx_1=\ld=\bx_K$, and the receiver applies a \gls{mrc}\footnote{\gls{rrharq} is also known as ``Chase combining'' \gls{harq}.} to combine the received signals $\by_1,\ld,\by_k$. We thus can use $h(x)=x$, and then $\SNR^\Sigma_k=\sum_{t=1}^k\SNR_k$ is the ``accumulated'' \gls{snr} \cite{Caire01}. In this case, \eqref{MD.2.1D} is exact, and the approximation sign may be removed.

In the case of \gls{irharq}, codewords $\bx_k$, $k=1,\ld, K$, are obtained by taking non-overlapping elements of  the ``mother'' codeword $\bx_\tr{o}=[\bx_1,\ld,\bx_K]$.\footnote{Codewords $\bx_k$ are \emph{punctured} versions of $\bx_\tr{o}$; for example, $\bx_k$ may consist of new parity symbols/bits in each \gls{harq} round.} As in \cite{Wan06, Park15,Alvarado15}, we will use $h(x)=I(x)$, which is concave so we guarantee that in the case of \gls{irharq} $\SNR^\Sigma_k>\sum_{t=1}^k\SNR_k$. That is, as expected, \gls{irharq} yields smaller probabilities of error compared to \gls{rrharq}. For simplicity, we use $I(\SNR)=\log_2(1+\SNR)$; although this corresponds to a Gaussian model for $S$, it is done merely to simplify the modelling. %and without any further assumptions about how the encoding is performed.

The final step requires finding the probability of the event $\nack_k$, and we use here the backward error implication assumption, $\err_k \Rightarrow \err_{k-1}\Rightarrow\ld\Rightarrow\err_1$, already proposed in \cite{Long09, Gu06}, which yields 
\begin{align}
\PR{\nack_k}&=\PR{\err_1,\ld,\err_k}
\label{NACK.ERR}
=\PR{\err_k}\approx\PER_k(\SNR_k^\Sigma).
\end{align}
We quickly note that, although formally $\PR{\nack_k}\leq\PR{\err_k}$, the approximation \eqref{NACK.ERR} is very accurate and should be preferred to the very imprecise $\PR{\nack_k}=\prod_{l=1}^k\PR{\err_l}$, which also appeared in the literature, \eg \cite{Zheng05, Lagrange10}.

\subsection{Relation to Previous Works}\label{Sec:state.of.art}

The interaction between the \gls{amc} and \gls{harq} has raised a considerable interest in the literature. For instance, the throughput of \gls{harq} with the \gls{amc} was analyzed in \cite{Liu04,Zheng05,Kim08b,Yu10,Ramis11,Zhang13}, and their delay in \cite{Le06,Wang07,Le07,Ishizaki07,Femenias09}. However, the available results do not allow us to draw clear-cut conclusions, mainly because they are based on different assumptions very often connected to a particular coding or channel models. 

The difficulty is to strike a balance between the simplicity of the analysis and the generality of the conclusions. To address this challenge, we use arguably the simplest non-trivial decoding error model, and consider extreme assumptions with regard to the channel model (slow and fast fading), while all the remaining cases (\eg correlated channels) are expected to yield intermediate results. Furthermore, we initially make no efforts to optimize the operation of \gls{harq}, as done for example in \cite{Ramis11}, providing thus a ``canonical" model for an interaction between the \gls{amc} and \gls{harq}.

However, the main difference with the previous works is that we adopt the throughput at \gls{phy} as the unique performance criterion, allowing the \gls{llc} to handle the residual errors of \gls{harq} or the \gls{amc}, at the cost of buffering and delay for the individual packets.

The conclusions (and the analysis) change when packet loss or delay are taken into consideration. For example, the authors of \cite{Liu04} analyze the \gls{llc} with \emph{truncated} the \gls{arq} operating on top of the \gls{amc}: under constraints on the probability of packet loss at \gls{llc}, increasing the number of the \gls{arq} rounds allows the \gls{amc} to select rather aggressively the rates providing hence a higher throughput. In our perspective, the loss is irrelevant and the highest throughput is attained by optimizing the \gls{amc} without any constraints, letting the \gls{llc} to deal with the errors.

%!TEX root =  ../HARQ-AMC.tex
\section{AMC}\label{Sec:AMC}

The throughput is a the long-term average number of correctly received bits per transmitted symbol. Since the errors in the \gls{amc} are block-wise memoryless, the throughput is unaffected by the fading type (slow or fast) and we can take expectation for each block
\begin{align}%\label{eta.amc.1}
\eta^\tnr{amc}
&\triangleq \Ex_{\SNRrv}\Big[\R_{\hat{l}}\big(1-\PER_{\hat{l}}(\SNRrv)\big)\Big]\label{eta.amc.2} \\
&=\sum_{l=1}^L \int_{\SNR\in\mcD_{l}}\pdf_{\SNRrv}(x) \eta_l(x) \dd x,
\end{align}
where we used \eqref{hat.l} to obtain \eqref{eta.amc.2}, with  
\begin{align}\label{eta.l.amc}
\eta_l(\SNR)\triangleq\R_l(1-\PER_l(\SNR)) 
\end{align}
being the ``instantaneous'' throughput defined for an \gls{snr}, $\SNR$. 

The optimal \gls{snr} decision region is hence given by 
\begin{align}\label{mcD.l}
\mcD^{\amc}_l=\set{\SNR: \eta_l(\SNR) \ge \eta_k(\SNR), \quad \forall k\neq l}.
\end{align}

\begin{proposition}\label{Prop:Intersection}
If the \gls{per} function is defined by \eqref{PER.SNR}, the optimal decision regions are intervals $\mcD^{\amc}_l=[\gamma_l, \gamma_{l+1})$.
\begin{proof}
It is enough to show that there exists $\SNR_\tnr{o}\ge\SNR_{\tr{th},l+1}$ such that $\eta_{l+1}(\SNR)<\eta_{l}(\SNR)$ for $\SNR\leq\SNR_\tnr{o}$, and $\eta'_{l+1}(\SNR)>\eta'_{l}(\SNR)$ for $\SNR\geq\SNR_\tnr{o}$. Then, it can be easily shown with basic algebra that there is only one value of $\SNR$ solving $\eta_{l+1}(\SNR)=\eta_{l}(\SNR)$. 
\end{proof}
\end{proposition}

Indeed, in most of the practically interesting cases, the decision regions of \gls{amc} are intervals \cite{Alouini00} with boundaries $\gamma_l$ defined by the intersection of $\eta_l(\SNR)$ and $\eta_{l-1}(\SNR)$
\begin{align}\label{kkt.eta}
&R_l\big(1-\PER_l(\gamma_l ) \big)
=R_{l-1}\big(1-\PER_{l-1}(\gamma_l ) \big),
\end{align}
where also, for notational convenience, we use $\gamma_0\triangleq0$ and $\gamma_{L+1}=\infty$.

This expression may be further simplified assuming that the probability of decoding error when transmitting with rate $R_{l-1}$ is very small at the right border of $\mcD^{\amc}_{l-1}$.\footnote{It is, indeed, the case if $\gamma_{l-1}\ll\gamma_l$ and the \gls{per} function decays quickly with $\SNR$ as per \eqref{PER.SNR}.} Then, \eqref{kkt.eta} becomes
\begin{align}\label{WEP.R.R}
\PER_{l} (\gamma_l)\approx 1-\frac{R_{l-1}}{R_l}, \quad l=2,\ld,L.
\end{align}

Here, we see that if $R_l$ is much larger than $R_{l-1}$, the nominal error rate at the interval border $\gamma_l$ may be quite high. For example, if $R_{l-1}=1$ and $R_{l}=2$, we obtain $\PER_{l} (\gamma_l)=0.5$. This result may be  contrasted with the ``hard'' limits imposed on the \gls{per}, $\PER_{l} (\gamma_l)=\PER_\tnr{t}$, suggested in~\cite{Liu04} and specified by the \gls{lte} as $\PER_\tnr{t}=10^{-1}$ \cite{3GPP_TS_36.213}.

Finally, we express the throughput compactly as follows:
\begin{align}\label{eta.amc}
\eta^\amc&=\sum_{l=1}^L R_l \cd (1-f_{1,l})p_l,
\end{align}
where 
\begin{align}\label{}
f_{1,l}&\triangleq \frac{1}{p_{l}}\int_{\gamma_l}^{\gamma_{l+1}} \pdf_{\SNRrv}(x) \PER_l(x)  \dd x,\quad
\label{p.l}
p_{l}\triangleq \int_{\gamma_l}^{\gamma_{l+1}} \pdf_{\SNRrv}(x) \dd x.
\end{align}

Using \eqref{WEP.R.R} in \eqref{Delta.epsilon}, we can calculate the border of the decision region as
\begin{align}
\label{gamma.l.Delta}
\gamma_l &=\SNR_{\tnr{th},l}\Big(1+ \frac{1}{\tilde{a}}\ln \frac{R_l}{\R_l-\R_{l-1}}\Big).
\end{align}
Therefore, not surprisingly, an increased strength of the coding (large $\tilde{a}$) moves the borders of the decision region closer to the decoding threshold, while a weak coding will result in larger $\gamma_l$; this, via \eqref{p.l}, will also move the throughput curve to the regions of higher \gls{snr}.

%%%%%%%%%%%%%%%%%%%%%
\subsection{Discussion}
We note that, according to the observations in \secref{Sec:LLC}, the above results hold also when \gls{arq} is used on top of the \gls{amc}. Since a similar setup is considered in \cite{Liu04}, it is instructive to contrast both results. 

The main difference is that, according to our results, the throughput-optimal intervals defined by $\gamma_l$ do not change with $M$ the number of \gls{arq} rounds. This is exactly the strength of the throughput criterion which does not change when \gls{arq} is added on top of a particular \gls{phy} transmission strategy. On the other hand, \cite{Liu04} imposed limits on the packet loss $P_{\tnr{loss}}$, and the \gls{snr} thresholds $\check{\gamma}_l$ were adjusted by solving the equation $[\PER_l(\check{\gamma}_l)]^M=P_{\tnr{loss}}$, or 
\begin{align}\label{gamma.check}
\PER_l(\check{\gamma}_l)=P_{\tnr{t}}, 
\end{align}
where  $P_\tnr{t}=P^{1/M}_\tnr{loss}$ is the target \gls{per} at each interval boundary $\check{\gamma}_l$. Hence, for a small $P_\tnr{loss}$, with a growing $M$ it is possible to increase $P_\tnr{t}$ and then $\check{\gamma}_l$ decreases approaching the throughput-optimal value $\gamma_l$; this effect was also observed in \cite{Liu04}. However, a caution is required because $\check{\gamma}_l$ should never be smaller than $\gamma_l$, which may happen using \eqref{gamma.check} for large $M$.

%%%%%%%%%%%%%%%%%%%%%
\begin{example}\label{Ex:AMC}
To illustrate the importance of the decision regions, we show in \figref{Fig:th_comp_AMC} the throughput of \gls{amc} for the (optimal) decision intervals based on \eqref{WEP.R.R} (as they are practically the same as with \eqref{kkt.eta}), and based on \eqref{gamma.check}, for both cases $P_\tnr{t}=10^{-1}$ and $P_\tnr{t}=10^{-2}$.

We use $L=5$ rates from the set $\mcR=\set{R_l}_{l=1}^L$, where $R_l=l\cd0.75$. From \eqref{WEP.R.R}, we find that $\PER(\gamma_l)\approx 0.5, 0.33, 0.25, 0.2$. To get an insight into the importance of the coding strength, we use $\tilde{a}\in\set{4, 0.5}$, where the weaker code throughput curve is right-shifted by $\sim6\tr{dB}$ due to the increased borders $\gamma_l$ of the decision regions, cf.~\eqref{gamma.l.Delta}. 

Quite clearly, the strict constraints on the packet loss (small $P_\tnr{t}$) move the \gls{snr} boundaries $\check{\gamma}_l$ to the right; this decreases the throughput and is particularly notable for weak codes (here, $\tilde{a}=0.5$). Adding more \gls{arq} rounds (increasing $M$) increases $P_\tnr{t}$ and moves $\gamma_l$ to the left (\cf \eqref{gamma.check}), hence improving the throughput which approaches the optimal solution derived in~\eqref{WEP.R.R}.

\begin{figure}[bt]
\psfrag{xlabel}[ct][ct][\siz]{$\SNRav$~[dB] }
\psfrag{ylabel}[ct][ct][\siz]{$\eta^\tnr{amc}$}
\psfrag{OPT}[lc][lc][\siz]{Eq.~\eqref{WEP.R.R}}
\psfrag{Ploss1XX}[lc][lc][\siz]{$P_\tnr{t}=10^{-1}$}
\psfrag{Ploss2XX}[lc][lc][\siz]{$P_\tnr{t}=10^{-2}$}
\psfrag{G=0.5}[lc][lc][\siz]{$\tilde{a}=0.5$}
\psfrag{G=1}[lc][lc][\siz]{$\tilde{a}=1$}
\psfrag{G=4}[lc][lc][\siz]{$\tilde{a}=4$}

\begin{center}
\scalebox{\sizf}{\includegraphics[width=\linewidth]{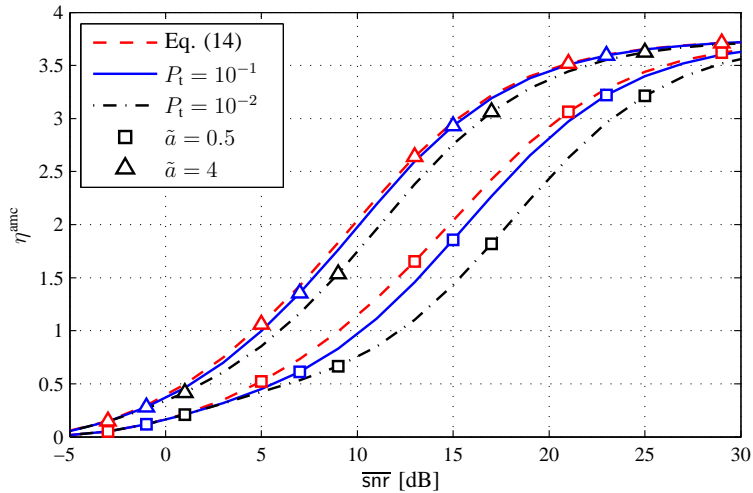}}
%\scalebox{\scalefig}{a)}\\
%\scalebox{\sizf}{\includegraphics[width=\linewidth]{./figures/HARQ_AMC_SLOW_multi}}\\
%\scalebox{\scalefig}{b)}
\caption{The throughput $\eta^\tnr{amc}$ with the decision intervals based on \eqref{WEP.R.R} and \eqref{gamma.check}, for $P_\tnr{t}=10^{-1}$ and $P_\tnr{t}=10^{-2}$.}\label{Fig:th_comp_AMC}
\end{center}
\end{figure}

\end{example}

%!TEX root =  ../HARQ-AMC.tex
\section{HARQ: Slow Fading}\label{Sec:slow.fading}
In the case of slow-fading channels, we may assume that the \gls{snr}, $\SNR$, remains constant during many blocks $\bs[n]$, which we treat together as a ``super channel-block''.  We may then apply the same approach we already used to analyze \gls{amc}, and calculate the throughput as 
\begin{align}\label{eta.harq.slow}
\eta^\harq_{K}&=\sum_{l=1}^L \int_{x \in\mcD_{l}}\pdf_{\SNRrv}(x) \eta^{\harq}_{K,l}(x) \dd x,
\end{align}
where again $\mcD_{l}$ is the \gls{snr} decision region in which we use the transmission rate $\R_l$, and $ \eta^{\harq}_{K,l}(x)$ is the throughput of \gls{harq} carried out with the rate $\R_l$ over \gls{snr} $x$. The main questions are: how to calculate $\eta^{\harq}_{K,l}(x)$, and how to find the optimal decision regions $\mcD_l$?

As to the throughput, $\eta^{\harq}_{K,l}(x)$, the number of channel blocks used to transmit the information packet is now variable due to the decoding errors whose probability is captured by the \gls{per} function in~\eqref{PER.SNR}. Modeling \gls{harq} rounds by the states of a Markov process, the beginning of the \gls{harq} cycle corresponds to the renewal of the process. Thus, we can use the renewal-reward theorem to obtain \cite{Zorzi96,Caire01,Zheng05},\cite[Ch.~2.3]{Wolff89_Book}
\begin{align}\label{eta.renew.reward.slow}
\eta^{\harq}_{K,l}(\SNR)=\frac{\Ex[\mfR_l]}{\Ex[\mfT_l]},
\end{align}
where $\mfR_l\in\set{\R_l,0}$ is a random ``reward'' (here, the number of successfully delivered bits normalized by the number of symbols, $\Ns$) at the end of the \gls{harq} cycle, and $\mfT_l\in\set{1,\ld,K}$ is a random inter-renewal time, \ie the number of \gls{harq} rounds.  

The expectations in \eqref{eta.renew.reward.slow} are taken with respect to the decoding errors $\err_1, \ld, \err_K$ which, conditioned on  $\SNRrv=\SNR$, are independent from one \gls{harq} cycle to another. Thus, the reward $\mfR_l$ and the time $\mfT_l$ are also independent between cycles. This is, in fact, the necessary condition to use the renewal-reward theorem~\eqref{eta.renew.reward.slow} \cite[Ch.~2.3]{Wolff89_Book}.

On the other hand, $\mfR_l$ and the time $\mfT_l$ are not \emph{unconditionally} independent, because they depend on the same realization of the \gls{snr}. For this reason, the expectation cannot be taken with respect to $\SNRrv$ as it would imply that the \gls{snr} is independent between the \gls{harq} cycles; while this approach may be suitable in bursty communication scenarios~\cite[Sec.~IV.B]{Shen09}\cite[Sec.~IV]{Makki14}, it is inappropriate in our case.\footnote{A direct consequence of bursty communications assumption is that the transmissions over ``better'' channels require shorter times than those made over ``worse'' channels. This creates a coupling between the channel model and the transmission protocol \cite{Makki14} which is undesirable in our communication model.}

The throughput can be hence calculated as \cite{Caire01,Zheng05}
\begin{align}\label{eta.harq.slow.l}
\eta^{\harq}_{K,l}(\SNR)=\frac{\R_l(1-f_{K,l}(\SNR))}{1+\sum_{k=1}^{K-1} f_{k,l}(\SNR)},
\end{align}
where
\begin{align}\label{}
f_{k,l}(\SNR)=\PR{\err_1,\ld,\err_k| \R_l}
\end{align}
is the probability of $k$ consecutive decoding errors conditioned on having the first \gls{harq} round carried out with the rate $\R_l$. These are calculated using the approach explained in~\eqref{NACK.ERR}
\begin{align}\label{}
f_{k,l}(\SNR)\approx\PER_l (\SNR_k^\Sigma),
\end{align}
where $\SNR_k^\Sigma$ is obtained from \eqref{aggreg.SNR}, and thus depends on the type of \gls{harq}. In particular, for \gls{rrharq} we have $\SNR_k^\Sigma=k\SNR$, while for \gls{irharq} it is obtained through $I(\SNR_k^\Sigma)=kI(\SNR)$, where $I(x)=\log_2(1+x)$ (see also \secref{Sec:Errors}). For notational convenience, we use $f_{0,l}(\SNR)\triangleq1$.

Before progressing any further, we first address the following simple question: does the throughput of \gls{harq} increase with then number of allowed transmissions, $\kmax$?

%%%%%%%%%%%%%%%%%%%%%%%%
\begin{proposition}\label{Thm22}
If 
\begin{align}\label{fk.cond.sup}
 \frac{f_{k+1,l}(\SNR)}{f_{k,l}(\SNR)} \leq \frac{f_{k,l}(\SNR)}{f_{k-1,l}(\SNR)}, \quad \forall k\ge 1
\end{align}
then, for $\eta^{\harq}_{K,l}(\SNR)$ given by \eqref{eta.harq.slow.l}, the following holds: $\eta^{\harq}_{K,l}(\SNR)\leq\eta^{\harq}_{K+1,l}(\SNR)$. Moreover, if there exists $k$ for which the inequality in~\eqref{fk.cond.sup} is strict, then we obtain a strict inequality $\eta^{\harq}_{K,l}(\SNR)<\eta^{\harq}_{K+1,l}(\SNR)$.
\begin{proof}

Our objective is to derive a sufficient condition, under which the expression
\begin{align}\label{}
\eta_K=\frac{R(1-f_{K})}{1+\sum_{k=1}^{K-1} f_{k}}
\end{align}
is increasing with $K$, where the dependence on $l$ and $\SNR$ is removed from \eqref{eta.harq.slow.l} for notational brevity. That is we require
\begin{align}\label{eta.harq.appendix}
\eta_{K}\leq\eta_{K+1}&\implies
\frac{R\big(1-f_{K}\big)}{1+\sum_{k=1}^{K-1} f_{k}} \le\frac{R\big(1-f_{K+1}\big)}{1+\sum_{k=1}^{K} f_{k}}\\
&\implies f_{K+1}\big(1 + \sum_{k=1}^{K-1}  f_{k}\big) \le f_{K}\sum_{k=1}^{K}  f_{k}\implies 
\label{final.cond}
\sum_{k=1}^{K} f_{K+1} f_{k-1} \le \sum_{k=1}^{K} f_{K} f_{k}.
\end{align}

Obviously, \eqref{final.cond} is a necessary and sufficient condition for the throughput to be increasing with $K$. To remove the dependence on the sum, we strengthen \eqref{final.cond} by imposing the inequality on each term yielding the sufficient (not necessary) condition
$f_{K+1} f_{k-1}\le f_{K} f_{k}\quad  \forall k\le K$, which must hold 
%that is
%\begin{align}\label{eq:condition1}
%\frac{f_{K+1}}{f_{K}} \le \frac{f_{k}}{f_{k-1} }\quad  \forall k\le K.
%\end{align}
%
$\forall K \ge 1$; this is \eqref{fk.cond.sup}. Any strict inequality will make the inequality in \eqref{eta.harq.appendix} strict as well. This terminates the proof.
\end{proof}
\end{proposition}

The above proposition rectifies the conditions exposed in~\cite[Sec.~IV]{Caire01}, where the ``subgeometric'' property $f_{k,l}< f_{1,l}^k$ was conjectured as a  sufficient condition for the throughput to be increasing with $K$. In fact, it is not, as we may show by constructing a subgeometric relationship and yet obtaining a non-increasing throughput with $K$. For example: the subgeometric terms $f_{2}=0.5f_{1}^2, f_{3}=0.75f_{1}^3$ yield $\eta^{\harq}_{2}>\eta^{\harq}_{3}$, which invalidates the conjecture of~\cite{Caire01}. 

The value of \propref{Thm22} is that, if the error probabilities in \gls{harq} satisfy the conditions \eqref{fk.cond.sup}, we can guarantee that \gls{harq} will improve the throughput of \gls{amc} provided the same \gls{snr} decision regions are used, \ie $\mcD_l=\mcD^{\amc}_l$. 

It is easy to show that the conditions in \propref{Thm22} are satisfied using \gls{rrharq} and \gls{irharq} under the decoding model of \secref{Sec:Errors}. %; this is also shown in \cite{Caire01} in the idealized setting of $\tilde{a}=\infty$.  
Thus, \gls{harq} will outperform the \gls{amc} in the slow-fading channels, and the gains may be larger if the decision regions \eqref{mcD.l} take into account the reality of \gls{harq} and are redefined as follows
\begin{align}\label{mcD.l.harq}
\mcD^{\harq}_l=\set{\SNR: \eta^{\harq}_{K,l}(\SNR) \ge \eta^{\harq}_{K,k}(\SNR), \quad \forall k\neq l}.
\end{align}

%%%%%%%%%%%%%%%%%%%%%%%%%%%%
\begin{figure}[bt]
\psfrag{xlabel}[ct][ct][\siz]{$\SNR$~[dB] }
\psfrag{ylabel}[ct][ct][\siz]{$\eta^{\harq}_{K,l}(\SNR)$}
\psfrag{R-1}[lc][lc][\siz]{$R_4$}
\psfrag{R-2}[lc][lc][\siz]{$R_5$}
\psfrag{etaAMC1}[rc][cc][\siz]{$\eta_4(\SNR)$~~}
\psfrag{etaAMC2}[lc][lc][\siz]{$\eta_5(\SNR)$}
\psfrag{etaHARQ1}[lc][lc][\siz]{$\eta^{\harq}_{K,5}(\SNR)$}
\psfrag{etaHARQ2}[lc][lc][\siz]{$\eta^{\harq}_{K,4}(\SNR)$}
\begin{center}
\scalebox{\sizf}{\includegraphics[width=\linewidth]{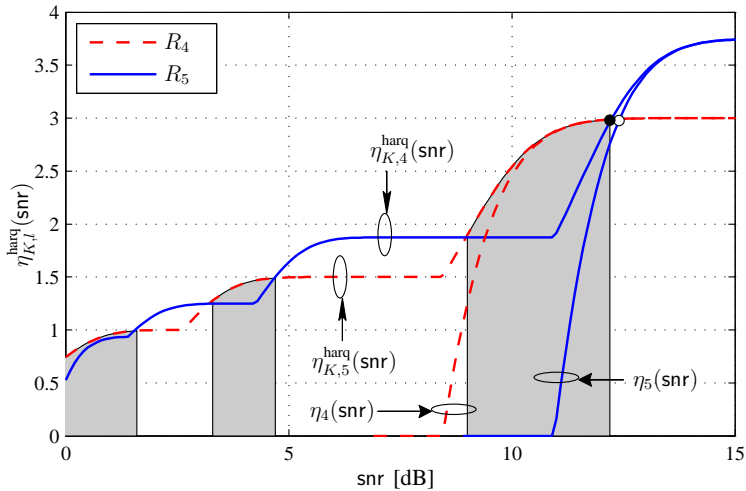}}
%\scalebox{\scalefig}{a)}\\
%\scalebox{\sizf}{\includegraphics[width=\linewidth]{./figures/HARQ_AMC_SLOW_multi}}\\
%\scalebox{\scalefig}{b)}
\caption{In slow-fading channels, decision regions $\mcD^\harq_l$ are unions of intervals. Here, $R_4=3$, $R_5=3.75$, and we assume \gls{irharq} with $K=4$. $\tilde{a}=4$.}\label{Fig:th_comp_slow}
\end{center}
\end{figure}
%%%%%%%%%%%%%%%%%%%%%%%%%%%%

Here, however, we may run into a difficulty because the throughput $\eta^{\harq}_{K,l}(\SNR)$ is not a concave function of $\SNR$. Since then, the unique intersection condition we obtained in~\propref{Prop:Intersection} cannot be guaranteed, the decision regions $\mcD^{\harq}_l$, in general, are non-convex sets and should be represented by the union of intervals as shown in the following example. %~\exref{Ex:2.eta.harq.slow}.

\begin{example}\label{Ex:2.eta.harq.slow}
With two rates $R_5=3.75$ and $R_4=3$ -- as in~\exref{Ex:AMC}, and \gls{irharq} with $K=4$, we can see from~\figref{Fig:th_comp_slow} that each decision region $\mcD^\harq_4$ (shaded) and $\mcD^\harq_5$ (white) is the union of intervals. The boundary of the right-most piece of $\mcD^\harq_5$ is shown with a black marker, where we can observe it being close to the boundary of $\mcD^\amc_5$ (shown with a white marker). Moreover, in most of the practical cases the rates are relatively close to each other, then if $R_l>R_{l-1}/2$ is ensured, it means that the boundary of $\mcD^\harq_l$ will be very close to the boundary of $\mcD^\amc_l$. Furthermore, if sufficiently ``dense" set of rates $\mcR$ is available, the region of the achievable throughput is well covered.
\end{example}

While, in theory,  the system may operate with arbitrary $\mcD^{\harq}_l$, it is much more convenient to use intervals as in the case of \gls{amc}. In fact, using $\mcD^\amc_l$ is not a bad idea since \propref{Thm22} confirms that the throughput will increase when using \gls{harq} on top of the \gls{amc}.

%%%%%%%%%%%%%%%%%%%%%%%%%%%%
\begin{figure}[bt!]
\psfrag{xlabel}[ct][ct][\siz]{$\SNRav$ [dB] }
\psfrag{ylabel}[ct][ct][\siz]{Throughput}
%\psfrag{g=0.5}[bl][bl][\siz]{$\tilde{a}=0.5$}
\psfrag{G=0.5}[lc][lc][\siz]{$\tilde{a}=0.5$}
%\psfrag{g=1}[bl][bl][\siz]{$\tilde{a}=1$}
%\psfrag{g=4}[bl][bl][\siz]{$\tilde{a}=4$}
\psfrag{G=4}[cl][cl][\siz]{$\tilde{a}=4$}
\psfrag{AMC}[bl][bl][\siz]{AMC}
\psfrag{HARQ-RR-xxx}[bl][bl][\siz]{\gls{rrharq}, $\mcD^\harq_l$}
\psfrag{HARQ-IR-opt}[bl][bl][\siz]{\gls{irharq}, $\mcD^\harq_l$}
\psfrag{HARQ-IR-AMC}[bl][bl][\siz]{\gls{irharq}, $\mcD^\amc_l$}
\begin{center}
\scalebox{\sizf}{\includegraphics[width=\linewidth]{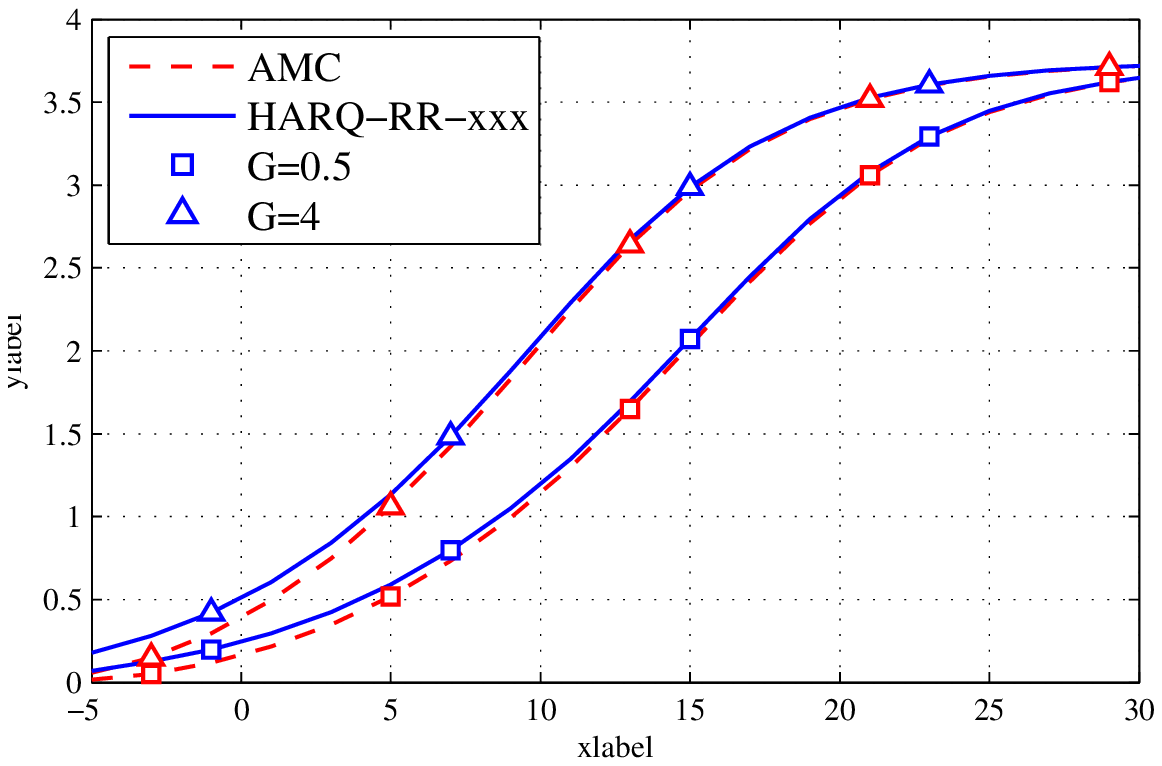}}\\
\scalebox{\siz}{a)}\\
\scalebox{\sizf}{\includegraphics[width=\linewidth]{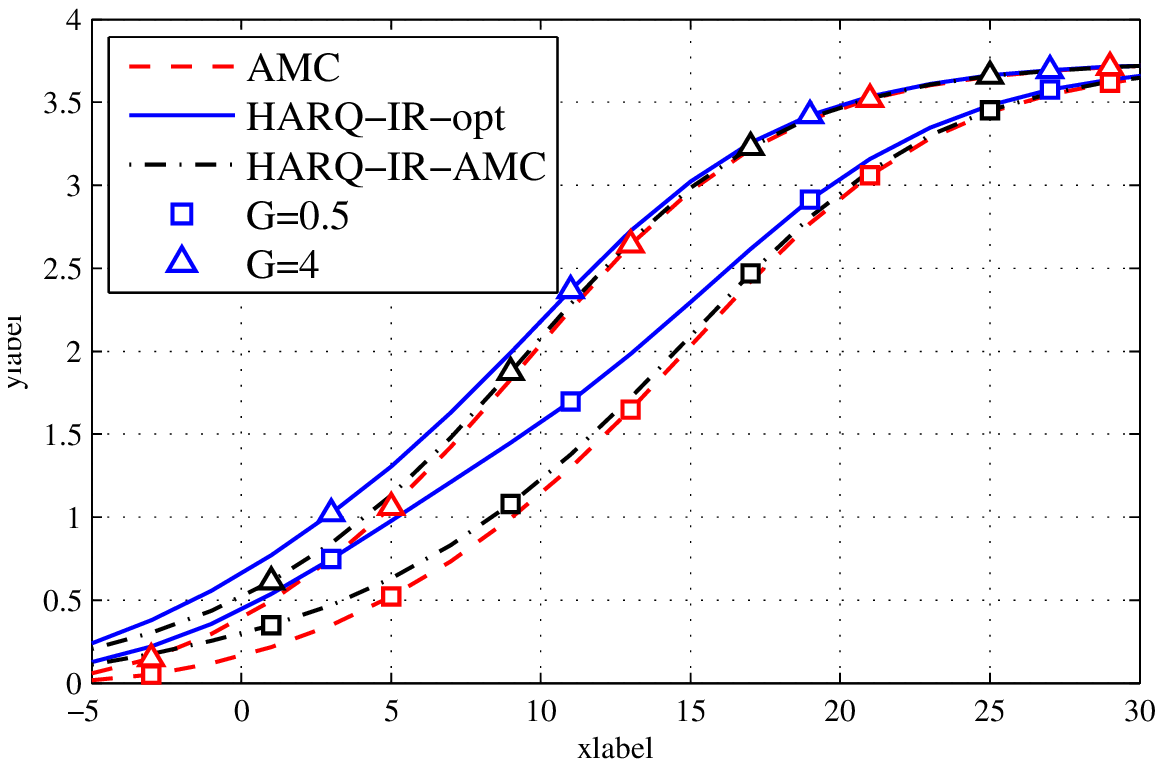}}\\
\scalebox{\siz}{b)}
\caption{Throughput of \gls{amc} and \gls{amc}-\gls{harq} for slow-fading channels using a)~\gls{rrharq}, and b)~\gls{irharq}.}\label{Fig:Slow.fading}
\end{center}
\end{figure}
%%%%%%%%%%%%%%%%%%%%%%%%%%%%

\begin{example}
In~\figref{Fig:Slow.fading}, we compare the throughput obtained using the \gls{amc} and \gls{harq} with rates from the set $\mcR$ as in \exref{Ex:AMC}. Both \gls{irharq} and \gls{rrharq} are considered with $K=4$ rounds, where we also consider two cases for the decision regions for \gls{harq}, namely, $\mcD^\amc_l$ and $\mcD^\harq_l$. Two coding strengths are considered: $\tilde{a}\in\set{0.5, 4}$. As expected, the throughput of \gls{harq} is improved with respect to the \gls{amc} even if ${\mcD}^{\amc}_{l}$ is used. The gains of using the optimal decision regions $\mcD^{\harq}_{l}$ are important only for \gls{irharq}; they are notable in low \gls{snr} and are pronounced for weak codes. On the other hand, in the case of \gls{rrharq}, we obtain $\mcD^{\harq}_{l}\approx\mcD^{\amc}_{l}$ and that is why we do not show two distinct curves in this case.

%Again, HARQ-IR gives satisfactory results when $\tilde{g}$ is small. We show also the result of AMC-HARQ-IR when the constraint of decodable first transmission is respected (AMC-AHRQ-IR-sub). The results show that this constraint decreases significantly the throughput of AMC-HARQ-IR with respect to AMC-AHRQ-IR-opt and HARQ-IR when $\tilde{g}$ is small. However, decodable first transmission does not affect the throughput of AMC-HARQ-RR as the optimal solution satisfies this constraint. The optimal thresholds for slow fading model are shown in Table~\ref{Table.thresholds} when $\tilde{g}=0.5$. AMC-HARQ-IR becomes more aggressive with respect to AMC and AMC-HARQ-RR. For example, a -10 dB shift is observed in $\gamma_3$ with AMC-HARQ-IR. 

%\begin{table}
%\begin{center}
% \begin{tabular}{|c c c c c|} 
% \hline
%   & $\gamma_2$ [dB]  & $\gamma_3$ [dB] &$ \gamma_4$ [dB] & $\gamma_5$ [dB] \\ [0.5ex] 
% \hline
% AMC & 6.39 & 10.79 & 14.21&17.20\\ 
% \hline
% AMC-HARQ-RR& 5.55 & 10.40 & 13.97&17.03 \\
% \hline
% AMC-HARQ-IR & -1.9382 &0.1703 &1.6137 &2.6482\\
% \hline
%\end{tabular}
%\end{center}
%\caption{Optimal thresholds in dB for slow fading model with $\tilde{g}=0.5$.}\label{Table.thresholds}
%\end{table}

\end{example}

%!TEX root =  ../HARQ-AMC.tex
\section{HARQ: Fast Fading}\label{Sec:Block.fading}

In the case of fast-fading channels, the throughput can again be expressed using the renewal-reward theorem \eqref{eta.renew.reward.slow}
\begin{align}\label{eta.renew.reward}
\eta^{\tnr{harq}}=\frac{\Ex[\mfR]}{\Ex[\mfT]},
\end{align}
where $\mfR$ is the random reward. Since the \gls{snr}s vary independently from one transmission round to another (and from one \gls{harq} cycle to another), the expectations in \eqref{eta.renew.reward} are taken also with respect to all the \gls{snr}s affecting the transmission (unlike in \eqref{eta.harq.slow}, where the expectation over $\SNRrv_1$ was taken outside of the fraction). Consequently, unlike~\eqref{eta.renew.reward.slow}, Since the rate is random from one \gls{harq} cykle to another, the expected reward is given by 
\begin{align}
\Ex[\mfR]&=\sum_{l=1}^L \R_l (1-f_{K,l}) p_l,
\end{align}
where $p_l$ is the probability of choosing the rate $\R_l$ in the first \gls{harq} round, given by
\begin{align}\label{}
p_l\triangleq \PR{\SNRrv \in \mcD_l},
\end{align}
with $\mcD^\harq_l$ being the decision regions we have to define.

Again, $f_{k,l}$ that denotes the probability of $k$ consecutive errors (conditioned on starting the \gls{harq} cycle with the rate $\R_l$) is given by
\begin{align}
%\label{Pr.Error}
f_{k,l}&\triangleq\PR{ \nack_k | \R_l }
\label{Pr.Error.2}
=\frac{1}{p_l}\int_{\mcD_l}
\pdf_{\SNRrv_1}(x)f_{k,l}(x)\dd x,\\
\label{f.klx}
f_{k,l}(x)
&=\Ex_{\set{\SNRrv_l}_{l=2}^k }
\big[\PER_l(x,\SNRrv_2,\ld,\SNRrv_k) \big].
\end{align}

Both \eqref{Pr.Error.2} and \eqref{f.klx} require a one-dimensional integration respectively over $\SNRrv_1\in\mcD_l$ and an aggregate \gls{snr} (merging $\SNR_1=x$ with $\SNRrv_2,\ld, \SNRrv_k$ via~\eqref{aggreg.SNR}).

Similarly, the expected number of~\gls{harq} rounds is given by 
\begin{align}\label{}
\Ex[\mfT]&=\sum_{l=1}^L \int_{\mcD_l}\pdf_{\SNRrv}(x) T_{K,l}(x) \dd x
=\sum_{l=1}^L\ov{T}_{K,l} p_l,
\end{align}
where $T_{K,l}(\SNR)=1+\sum_{k=1}^{K-1} f_{k,l}(\SNR)$ is the expected number of rounds with the first being carried out over $\SNR\in\mcD^\harq_l$, and 
$\ov{T}_{K,l}=1+\sum_{k=1}^{K-1} f_{k,l}$.

%%%%%%%%%%%%%%%%%%%%%%%%%%%%%%
The throughput of \gls{harq} is then given by
\begin{align}\label{eta.harq}
\eta^{\tnr{harq}}_K=\frac{\sum_{l=1}^L R_l (1-f_{K,l})p_l}{\sum_{l=1}^L  \ov{T}_{K,l}p_l },
\end{align}
where both the numerator and the denominator depend on $\mcD^\harq_l$, $l=1, \ld, L$.

For the moment, we do not make any assumption about the coding, \ie we make no distinction between \gls{irharq} and \gls{rrharq}. 

%%%%%%%%%%%%%%%%%%%%%%%%%%%%%%
\subsection{Decision Regions}
The regions $\mcD^\harq_l$ that maximize~\eqref{eta.harq} are in general non-convex and must by represented as a union of intervals. Since this would lead to a tedious optimization problem, we restrict our considerations to the \gls{amc}-like decision regions 
\begin{align}\label{}
\tilde{\mcD}^\harq_l=[\gamma_l,\gamma_{l+1}),
\end{align}
which, as we will see, become an insightful proxy to the optimal solution despite their suboptimality.

Denoting \eqref{eta.harq} as $\eta^{\tnr{harq}}_K(\boldsymbol{\gamma})$, where $\boldsymbol{\gamma}\triangleq[\gamma_1,\ld, \gamma_L]$, we solve the following optimization problem
\begin{align}\label{opt_prob}
\hat{\eta}^{\tnr{harq}}_K\triangleq &\max_{\boldsymbol{\gamma}} \eta^{\tnr{harq}}_K(\boldsymbol{\gamma})  
\quad \text{s.t.}\quad \gamma_1\le \gamma_2 \le\ld\le\gamma_L\le\gamma_{L+1},
\end{align}
where $\gamma_1=0, \gamma_{L+1}=\infty$ are adopted for notational convenience.

Since the numerator and the denominator of \eqref{eta.harq} depend on $\boldsymbol{\gamma}$, we will use the fractional programming approach \cite[Proposition~2.1]{crouzeix1991}. We define the following function
\begin{align}\label{lagrangian}
F(\boldsymbol{\gamma},\lambda)\triangleq \sum_{l=1}^L R_l &p_l\big(1-f_{K,l}\big)
-\lambda \big (\sum_{l=1}^L  \ov{T}_{K,l} p_l\big),
\end{align}
and solve the new optimization problem
 \begin{align}\label{opt_prob1}
\boldsymbol{\gamma}_{\lambda}\triangleq &\argmax_{\boldsymbol{\gamma}} F (\boldsymbol{\gamma},\lambda) 
\quad\text{s.t.} \quad 0=\gamma_1\le \gamma_2 \le\ld\le\gamma_L\le\infty.
\end{align}

Then, the optimal solution of \eqref{opt_prob} is found examining the sign of $F(\boldsymbol{\gamma},\lambda)$ due to the simple relationships \cite{crouzeix1991}
\begin{align}
 F(\boldsymbol{\gamma}_{\lambda},\lambda)>0 & \iff \lambda<\hat{\eta}^{\tnr{harq}}_K, \\
F(\boldsymbol{\gamma}_{\lambda},\lambda)<0 & \iff \lambda>\hat{\eta}^{\tnr{harq}}_K,  \\
 F(\boldsymbol{\gamma}_{\lambda},\lambda)=0 & \iff \lambda=\hat{\eta}^{\tnr{harq}}_K  .
\end{align}

Consequently, to solve \eqref{opt_prob}, it is enough to solve \eqref{opt_prob1} for each $\lambda$, and find $\lambda$ such that $F(\bx_{\lambda},\lambda)=0$. The latter can be done efficiently using simple numerical methods (a bisection search, for instance). To solve \eqref{opt_prob1}, we differentiate \eqref{lagrangian} with respect to $\gamma_l$, and obtain the following \gls{kkt} equations
\begin{align}\nonumber
R_{l-1}\big(1-\Ex_{\set{\SNRrv_l}_{l=2}^K }
\big[\PER_{l-1}(\SNR^\Sigma_K(\gamma_l)) \big]\big)
&-R_{l}\big(1-\Ex_{\set{\SNRrv_l}_{l=2}^K }
\big[\PER_l(\SNR^\Sigma_K(\gamma_l)) \big]\big)\\
\label{HARQ_opt1}
&\qquad=\lambda\big(
T_{K,l-1}(\gamma_l) -T_{K,l}(\gamma_l) 
\big),
\end{align}
that may be solved for each $\gamma_l$ under the constraints $\gamma_l\leq\gamma_{l+1}$, $l=1,\ld, L$. We note that it is possible to obtain $\gamma_l=\gamma_{l+1}$; in which case the decision region is degenerate $\tilde{\mcD}^\harq_l=\emptyset$. 

%%%%%%%%%%%%%%%%%%%%%%%%%%%%%%%%%%%%%%%%
\begin{figure}[ht]
\psfrag{xlabel}[ct][ct][\siz]{$\SNRav$~[dB] }
\psfrag{ylabel}[ct][ct][\siz]{$\gamma_l$ [dB]}
\psfrag{gamma2}[cl][cl][\siz]{$\gamma_2$}
\psfrag{gamma3}[cl][cl][\siz]{$\gamma_3$}
\psfrag{gamma4}[cl][cl][\siz]{$\gamma_4$}
\psfrag{gamma5}[cl][cl][\siz]{$\gamma_5$}
\psfrag{\gls{amc}}[cl][cl][\siz]{\gls{amc}}
\psfrag{\gls{amc}-HARQ-IR}[cl][cl][\siz]{\gls{irharq}}
\psfrag{\gls{amc}-HARQ-RR}[cl][cl][\siz]{\gls{rrharq}}
\begin{center}
\scalebox{\sizf}{\includegraphics[width=\linewidth]{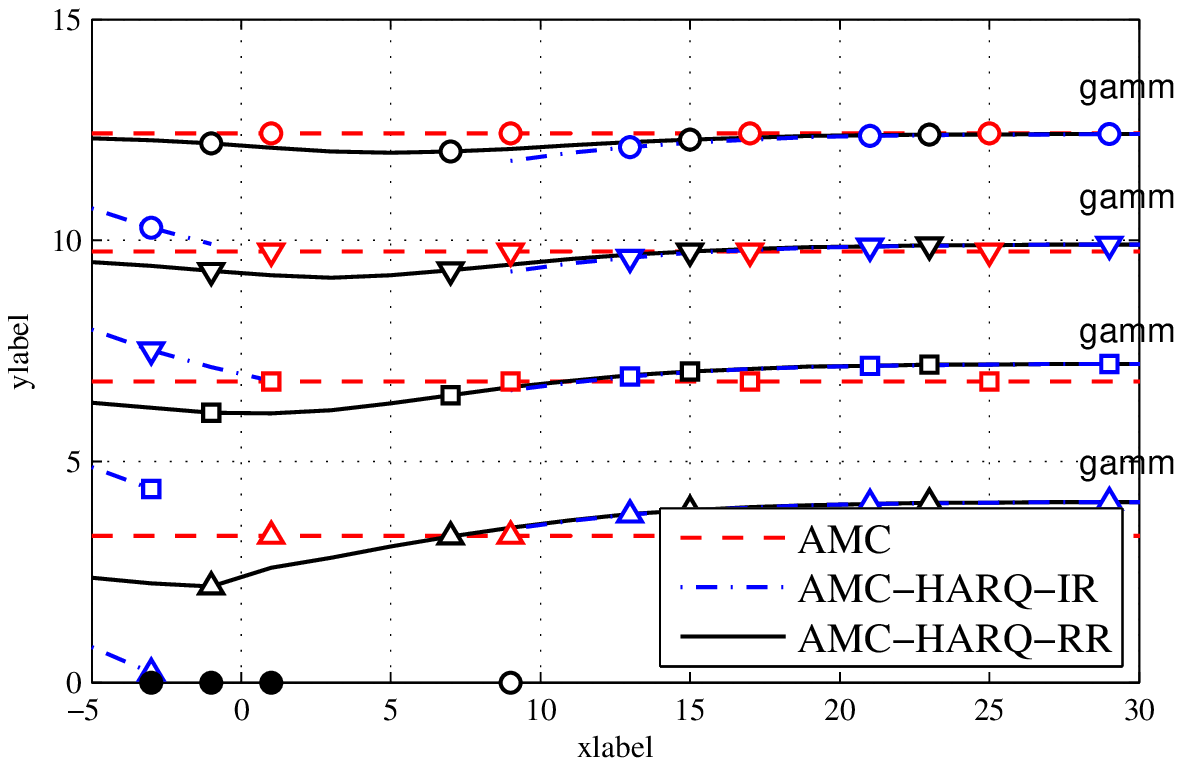}}\\
\scalebox{\siz}{a)}\\
\scalebox{\sizf}{\includegraphics[width=\linewidth]{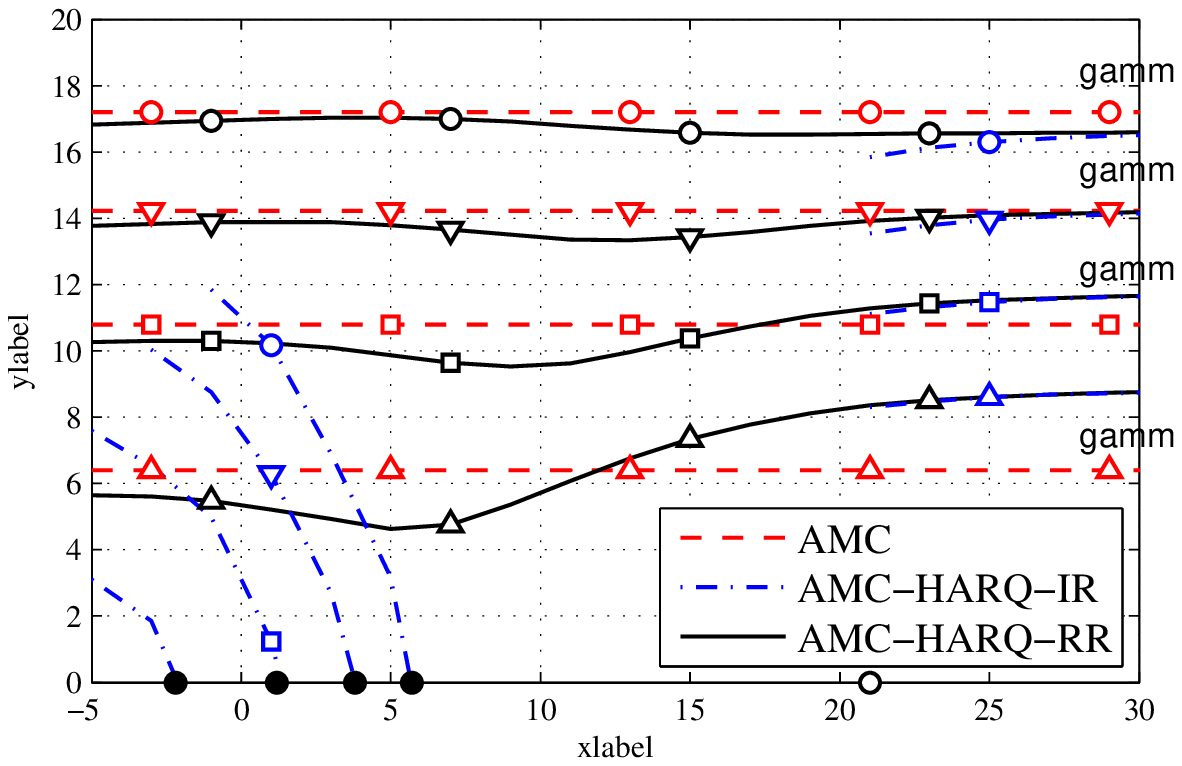}}\\
\scalebox{\siz}{b)}
\caption{Optimal thresholds $\gamma_l$ defining the decision regions $\tilde{\mcD}^\harq_l$ for fast-fading channels with a)~$\tilde{a}=4$, and b)~$\tilde{a}=0.5$. The filled circles on the $\SNRav$ axis indicate the value $\SNRav_{\tnr{low},l}$ such that if $\SNRav<\SNRav_{\tnr{low},l}$; $\tilde{\mcD}^\harq_l$ is not degenerate, \ie $\gamma_l>\gamma_{l-1}$. We note that in a)~we have $\SNRav_{\tnr{low},4}=\SNRav_{\tnr{low},5}$. The hollow circle indicates the value $\SNRav_{\tnr{high}}$ (common to all decision regions) such that if $\SNRav>\SNRav_{\tnr{high}}$; $\tilde{\mcD}^\harq_l$ is not degenerate.}\label{Fig:IID.Thresholds}
\end{center}
\end{figure}
%%%%%%%%%%%%%%%%%%%%%%%%%%%%%%%%%%%%%%%%
\begin{example} \label{Ex:Thresholds}
Considering the same set of rates $\mcR$ as in previous examples, we show in \figref{Fig:IID.Thresholds} the thresholds $\set{ \gamma_2,\gamma_3,\gamma_4,\gamma_5}$  obtained after solving \eqref{opt_prob}. The optimal thresholds for the \gls{amc}, which do not depend on $\SNRav$, are shown as a reference. 

The relationship between $\gamma_l$ and $\SNRav$ is quite complex, and it provides an immediate indication on how tedious it will be to find the optimal regions $\mcD^\harq_l$ (as a union of intervals). Some intuition about the potential gains of using \gls{harq} on top of \gls{amc} may be obtained from~\figref{Fig:IID.Thresholds}a. For high (average) \gls{snr}s, the thresholds of \gls{harq} are higher compared to the \gls{amc}, this more ``conservative" choice of rates is meant to avoid errors and, hence, not to use \gls{harq}. On the other hand, for low \gls{snr}s, the thresholds of \gls{harq} are lower compared to the \gls{amc}: the rates are adopted in a more ``aggressive" way indicating that the retransmissions would be beneficial. This becomes particularly clear with medium \gls{snr}s, where the regions degenerate and only $\tilde{\mcD}_5=[0,\infty)$; this occurs with \gls{irharq} for $\SNRav\in (1~\!\tnr{dB}, 9~\!\tnr{dB})$, and this interval is indicated with the markers on the $\SNRav$ axis. 

Similar conclusions may be drawn from~\figref{Fig:IID.Thresholds}b, where the relationship is slightly more complex due to a much less steeper \gls{per} curve of the decoder.

We also observe that, for high \gls{snr}s, all decision regions are not degenerate, and we conjecture that it would be always true. We adopt the following reasoning: we know that the region $\mcD^\harq_L$ is not degenerate because, if it was, it would not be possible to attain a throughput that is close to $R_L$. Now, for all $\SNR<\SNR_{\tnr{th},L}$, transmitting with the rate $\R_L$ provokes an error, so we ask a question: is it more beneficial to avoid errors or to count on retransmissions? This can be answered comparing the rewards, $\mfR$, associated with the selections of rates (actions). The assignment $\SNR\in\mcD^\harq_{L-1}$ yields an expected immediate reward of $\mfR=\R_{L-1}(1-f_{L-1,1})$; on the other hand, if $\SNR\in\mcD^\harq_L$, the reward is at most $\mfR=\R_L/2$ because it takes at least two \gls{harq} rounds. Thus, for $\R_{L-1}\gg \R_L/2$ and a sufficiently small $f_{l,1}$, we will assign $\SNR$ to $\mcD_l$. A similar reasoning holds when $\R_{L-1}\gg \R_L/k$. We thus conclude that the region $\mcD^\harq_{L-1}$ is not degenerate. Furthermore, with the same arguments for $\SNR<\SNR_{\tnr{th},L-1}$, we can find that the regions $\mcD^\harq_{L-2}$ is not degenerate.
\end{example}

\subsection{Throughput}

\propref{Thm22} answered the important question about the value of using \gls{harq} on top of the \gls{amc} in slow-fading channels. In fast-fading channels, however, the derivations are more complicated due to the sums appearing in the numerator and the denominator in~\eqref{eta.harq}. To obtain answers in this case, we will consider separately the regimes of high- and low- average \gls{snr}s. 

From the results in \exref{Ex:Thresholds}, we have already obtained an indication about the added-value of \gls{harq}: retransmissions were more valuable in the low- than in the high-\gls{snr} regime. However, the final conclusions will be drawn by looking at the throughput.

\begin{example}\label{Ex:Throughput.fast}
The throughput obtained by adopting the decision regions shown in~\exref{Ex:Thresholds} is represented in~\figref{Fig:IID.fading}. The results may seem surprising at the first glance, but they are completely in line with the behavior of the decision regions: the advantage of \gls{harq} is well pronounced in low $\SNRav$. On the other hand, for high $\SNRav$,  adding \gls{harq} on top of the \gls{amc} is counterproductive. The throughput penalty is small but clearly observable, \eg considering $\tilde{a}=4$ and \gls{rrharq}, this effect occurs above $3~\!\tnr{dB}$, and above $9~\!\tnr{dB}$ for \gls{irharq}; these break-points are indicated with markers on the figure.

We also note that using the optimized decision regions $\tilde{\mcD}^\harq_l$ is beneficial for \gls{amc}-\gls{irharq} when the break-points move to the high-\gls{snr}s. However, this gain is moderate for strong codes ($\tilde{a}=4$), and the general conclusion holds: there are regions of \gls{snr} for which it is counterproductive to use \gls{harq}.

Finally, \figref{Fig:IID.fading}b shows the throughput with \gls{irharq} in the range $\SNRav\in(1~\!\tnr{dB}, 9~\!\tnr{dB})$ for $\tilde{a}=4$, where $\tilde{\mcD}^\harq_5=[0, \infty)$ (cf.~\exref{Ex:Thresholds}) implies that we effectively ignore the \gls{csi}; meaning that the adaptation to the channel is counterproductive. This observation is explained by the fact that \gls{amc}-\gls{harq} cannot be considered as a fully adaptive transmission since it ignores the \gls{csi} after the first round. In fact, this is the source of a surprising and disappointing behavior of \gls{harq} in high-\gls{snr}s.
\end{example}

%%%%%%%%%%%%%%%%%%%%%%%%%%%%%%
\begin{figure}[ht]
\psfrag{xlabel}[ct][ct][\siz]{$\SNRav$ [dB] }
\psfrag{ylabel}[ct][ct][\siz]{Throughput}
\psfrag{G=0.5}[cl][cl][\siz]{$\tilde{a}=0.5$}
\psfrag{G=1}[cl][cl][\siz]{$\tilde{a}=1$}
\psfrag{G=4}[cl][cl][\siz]{$\tilde{a}=4$}
\psfrag{AMC}[cl][cl][\siz]{\gls{amc}}
\psfrag{AMC-HARQ-RR}[cl][cl][\siz]{\gls{rrharq}, $\tilde{\mcD}^\harq_l$}
\psfrag{HARQ-IR-opt}[cl][cl][\siz]{\gls{irharq}, $\tilde{\mcD}^\harq_l$}
\psfrag{HARQ-IR-amc}[cl][cl][\siz]{\gls{irharq}, $\mcD^\amc_l$}
\begin{center}
\scalebox{\sizf}{\includegraphics[width=\linewidth]{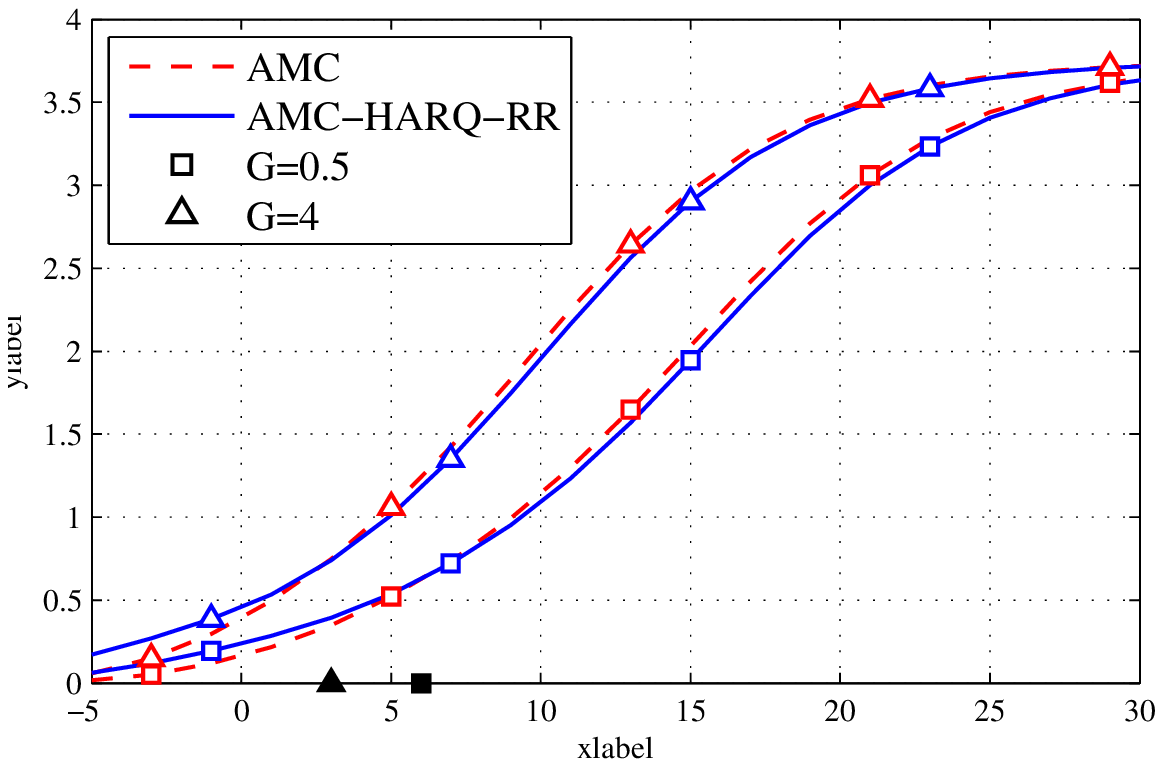}}\\
\scalebox{\siz}{a)}\\
\scalebox{\sizf}{\includegraphics[width=\linewidth]{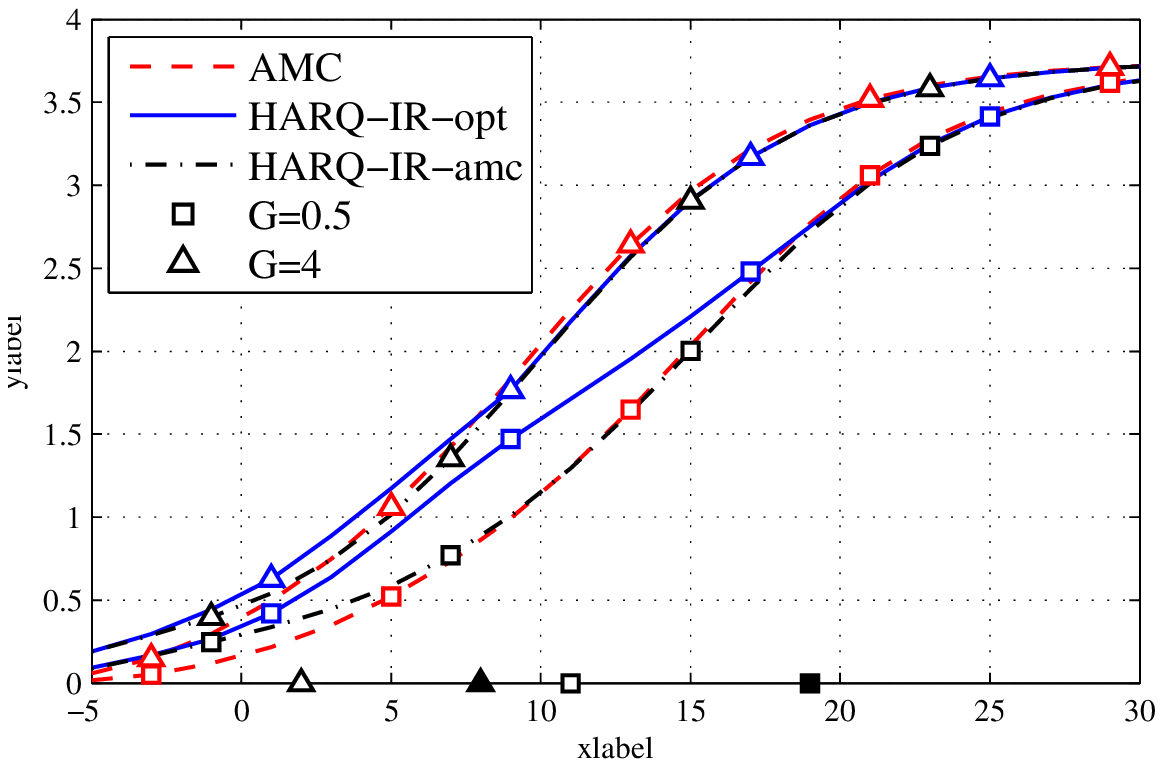}}\\
\scalebox{\siz}{b)}
\caption{The throughput of \gls{amc} and \gls{amc}-\gls{harq} over fast-fading channels for a)~\gls{rrharq} and b)~\gls{irharq}. The markers on the $\SNRav$ axis indicate the points where the \gls{harq} throughput curves cross the throughput curve of \gls{amc}. Filled and hollow markers correspond to \gls{harq} based on $\tilde{\mcD}^\harq_l$ and $\mcD^\amc_l$, respectively.}\label{Fig:IID.fading}
\end{center}
\end{figure}
%%%%%%%%%%%%%%%%%%%%%%%%%%%%%%

We will explain now more formally the results we have observed in the examples. We will show that, in the low-\gls{snr} regime, \gls{harq} is more productive, but at high $\SNRav$, \gls{amc} outperforms \gls{amc}-\gls{harq} under mild conditions on the \gls{per} curves of the decoder.

\begin{proposition}[Low SNR]\label{Prop:harq.wins}
Denote by $\hat{\eta}^\harq_K$ the optimal throughput of \gls{harq} based on $\mcD^\harq_l$. There exists $\SNRav_\tnr{o}$ such that, for all $\SNRav<\SNRav_\tnr{o}$,
\begin{align}\label{harq.wins}
\hat{\eta}^\harq_K \geq \eta^\amc.
\end{align}
\begin{proof}
We denote by $\check{\eta}^\harq_K$ the throughput of \gls{harq} based on $\mcD^\amc_l$ where, due to the non-optimized choice of the decision regions, we have $\check{\eta}^\harq_K\leq \hat{\eta}^\harq_K$. Since $\lim_{\SNRav\rightarrow 0}p_l f_{K,l}=0, l>1$, and $\lim_{\SNRav\rightarrow 0}f_{K,1}p_1=f_{K,1}$, we compare the throughputs using the terms appearing in the limits
\begin{align}
\label{lim0.harq}
\lim_{\SNRav\rightarrow 0}\check{\eta}^{\harq}_K&= \frac{R_1(1-f_{K,1})}{\ov{T}_{K,1}},\\
\label{lim0.amc}
\lim_{\SNRav\rightarrow 0}\eta^{\amc}&= R_1(1-f_{1,1}).
\end{align}
Thus, in the low-\gls{snr} regime, both \gls{harq} and the \gls{amc} may be considered as single-rate transmission schemes (\ie without rate adaptation). So, from \propref{Thm22}  we have $\check{\eta}^{\harq}_K \ge \eta^{\amc}$.
\end{proof}
\end{proposition}

For high $\SNRav$, the situation is slightly more involved, and we cannot use \eqref{lim0.harq} because all terms $f_{K,l}p_l$ have the same limit, $\lim_{\SNRav\rightarrow\infty}f_{K,l}p_l=0$. To compare the throughputs of \gls{harq} and \gls{amc}, we will thus define a \gls{trp} as a hypothetical \gls{harq} transmission which, in the first round, behaves just like a conventional \gls{harq} (and is described by the same probability of decoding error), while in the second transmission round, it guarantees that the message is decoded, \ie $f_{2,l}=0$. The throughput of \gls{trp} is given by 
\begin{align}\label{eta.bound}
\eta^{\trp} &\triangleq\frac{\sum_{l=1}^L R_l p_l}{1 + \sum_{l=1}^L  f_{1,l} p_l } \nonumber\\
&=\frac{\sum_{l=1}^L R_l p_l}{1 + \ov{f}_{1} },
\end{align}
where $\ov{f}_{1}\triangleq \sum_{l=1}^L  f_{1,l} p_l$ is the average probability of error in the first round, \ie the probability that re-transmissions are required, $\ov{f}_{1}=\PR{\nack_1}$.\footnote{This should not be confused with the average \gls{per}.} Comparing \eqref{eta.bound} and \eqref{eta.harq}, we see that the throughput of \gls{trp} upper bounds the throughput of \gls{harq}, $\eta^\harq_K <\eta^{\trp}$.

To focus the discussion, we start with the assumption that the decision regions $\mcD^\harq_l$ are not degenerate. 

%%%%%%%%%%%%
\begin{proposition}\label{Prop:amc.ge.harq} 
Let $x_{l}\triangleq \lim_{\SNRav\rightarrow\infty}\frac{f_{1,l}p_l}{f_{1,L}p_L}$. 

If  $x_l>0$, $l=1,\ld, L-1$, then there is a $\SNRav_\tnr{o}$ such that, for all $\SNRav>\SNRav_\tnr{o}$,
\begin{align}\label{harq.le.amc}
\eta^{\harq}_K < \eta^\amc.
\end{align}
\begin{proof}
Since $\eta^\harq_K \le\eta^{\trp}$, it is enough to prove that $\eta^{\trp}<\eta^{\amc}$ to obtain a sufficient condition for \eqref{harq.le.amc}. We have
\begin{align}\label{harq.amc.condition}
\eta^{\trp}&<\eta^{\amc} \\
\frac{\sum_{l=1}^L R_l p_l}{1 + \ov{f}_{1} }&<\sum_{l=1}^L R_l (1-f_{1,l})p_l\\
 \sum_{l=1}^L R_l p_l&<(1 + \ov{f}_{1})\sum_{l=1}^L R_l (1-f_{1,l})p_l\\
\sum_{l=1}^L R_l f_{1,l}p_l&<\ov{f}_{1}\eta^\amc,
\label{eta.harq.simple}
\end{align}
and by taking the limit $\SNRav\rightarrow\infty$ of both sides of \eqref{eta.harq.simple}, we get
\begin{align}
\sum_{l=1}^L  R_l x_{l,L} &< \sum_{l=1}^L R_L x_{l,L},
\end{align}
which is obviously true if $x_l>0$. Thus, for sufficiently large $\SNRav$, $\eta^{\amc}>\eta^{\trp}>\eta^\harq_K$.
\end{proof}
\end{proposition}

The value of \propref{Prop:amc.ge.harq} is that it is obtained using solely the model of the decoding errors in the first transmission. Hence, we  do not need to consider approximations with regard to the way the redundancy is introduced by \gls{harq}. We also quickly note that, using the model of the decoding errors from~\eqref{PER.SNR} and~\cite[Eq. (8)]{Liu04}, we obtain $x_{l,L}>0, \forall l$, provided that the decision regions $\mcD^\harq_l$ are not degenerate (\ie when $p_l>0$). 

This brings us to the following question: will we be able to obtain gains by using \gls{harq} on top of \gls{amc} if the decision regions $\mcD^\harq_l$ are optimized? The answer is ``no'', because \propref{Prop:amc.ge.harq} does not assume anything about the decision regions $\mcD^\harq_l$, except that there are at least two (in which case there are at least two non-zero terms $x_l$). This condition materializes for the high-\gls{snr} regime under the conjecture we made in \exref{Ex:Thresholds}.

Another interesting question is whether we can decrease the gap between $\eta^{\amc}$ and $\eta^{\trp}$ by increasing the number of rated in the set $\mcR$? The question may be difficult to answer. Nevertheless we consider here a particular and practical case where we cannot use coding rates below $\R_1$ which is determined by the minimum implementable encoding rate (\eg the coding rate of $1/3$ in the popular turbo-codes) and the minimum modulation rate (\eg \gls{bpsk}). Also, a finite size of the constellation, \eg $16$-\gls{qam}, implies that there is a maximum rate $\R_L$ which cannot be exceeded. We can manipulate, however, the granularity of the set $\mcR$, \ie we can increase the number of rates, $L$, while keeping $\R_1$ and $\R_L$ constant.

%%%%%%%
\begin{corollary}\label{Corr:gap}
For $\tilde{a}=\infty$, keeping $\R_1$ and $\R_L$ constant, the difference $ \eta^{\amc}-\eta^{\trp}$ increases with $L$.
\end{corollary}
\begin{proof}
We first rewrite \eqref{harq.amc.condition} as follows
\begin{align} 
 \eta^{\amc}-\eta^{\trp}= \frac{\ov{f}_{1}(\eta^{\amc}-R_1)}{1+\ov{f}_{1}}.
\end{align}
Then, for $\tilde{a}=\infty$, $\ov{f}_{1}$ is independent from $L$ because the only source of error is the event $\set{\SNRrv\in (0, \SNR_{\tnr{th},1} )}$. Since $\R_1$ is fixed by assumption, and $\eta^{\amc}$ is increasing with $L$, the difference $\eta^{\amc}-\eta^{\trp}$ can only increase with $L$.
\end{proof}

Although \corref{Corr:gap} does not talk about  the gap $\eta^{\amc}-\eta^{\harq}$, according to our observations, for high-\gls{snr} regime, we obtain a $\eta^{\harq}$  indistinguishable from $\eta^\trp$. Moreover, setting a practical (finite) value for $\tilde{a}$ does not change the results, and we indeed observe that the difference $\eta^{\amc}-\eta^{\harq}$ increases with $L$ even if $\tilde{a}<\infty$.

%%%%%%%%%%%%%%%%%%%%%%%%%
\subsection{Discussion}\label{Sub:discussion}

The presented results indicate that adding \gls{harq} on top of the \gls{amc} may be counterproductive in terms of throughput. This finding is not only counterintuitive; but goes against the general belief that both the \gls{amc} and \gls{harq} complement each other in their operations. We try to provide more intuition about this unusual---at first glance---behavior. We do it in two ways: first, by exploring the temporal relationships of \gls{csi}, and second, by reinterpreting the equation in the proof of~\propref{Prop:amc.ge.harq}. 

%%%%%
\subsubsection{Temporal evolution of the channel states}\label{Sec:Temporal}
Let us start with the relationship between the channel state in two different time instants assuming operation in high $\SNRav$. Then, the \gls{amc} is very likely to choose a high transmission rate, \eg $R_L$. A small rate, \eg  $R_1$, may be chosen as well, which happens when the measured \gls{snr}, $\SNR[n]$, falls into the interval $(0,\gamma_2)$; this event has a probability $p_1>0$. Let us assume that the error occurs when transmitting with the rate $R_1$, so the reward in time $n$ is zero, $\mfR[n]=0$. We can now compare the expected reward in the block $n+1$ due to the adoption of \gls{amc} or \gls{amc}-\gls{harq}.

With the \gls{amc}-\gls{harq}, the error is handled by the second \gls{harq} round where  %. The probability $\PR{\SNRrv[n+1]\in\mcD^\harq_l}=p_l$, $l>1$ is large, and so are 
the chances for a successful decoding (both for \gls{irharq} and \gls{rrharq}), $f_{1,2}$ are high; the expected reward in this second round is thus given by  
\begin{align}\label{reward.harq.2}
\Ex\big[\mfR[n+1]\big]=R_1(1-f_{1,2})\approx R_1.
\end{align}
On the other hand, with the \gls{amc}, the error is handled by the \gls{llc}, so the \emph{new} packet is transmitted in time $n+1$ with the rate $R_l$. The expected reward in this case with the \gls{amc} in time $n+1$ is thus 
\begin{align}\label{reward.amc.1}
\Ex\big[\mfR[n+1]\big]=\R_l(1-f_{l,1})>R_{l-1},
\end{align}
where the loose lower bound is obtained by using~\eqref{WEP.R.R} and the fact that $f_{l,1}<\PER_l(\gamma_l)$.

We can see that the reward \eqref{reward.harq.2} due to a small probability of error, $f_{1,2}\approx 0$ (obtained thanks to \gls{harq}), is negligible corresponding to the worst-case gain of the \gls{amc} given by \eqref{reward.amc.1}. In other words, even if the selection of the rate $R_1$ occurs with a very small probability $p_1$, its effect propagates to the following time instants.

%%%%%
\subsection{Probabilistic interpretation}\label{Sec:Proba}
Let us transform~\eqref{eta.harq.simple} and evaluate its limits
\begin{align}\label{}
\frac{\sum_{l=1}^L R_l f_{1,l}p_l}{\ov{f}_{1}}
=
\frac{\sum_{l=1}^L R_l \PR{\nack, \R_l}}{\PR{\nack}}%\\
\label{proba.ineq}
&=
\sum_{l=1}^L R_l \PR{\R_l|\nack}<\eta^{\amc},\\
\label{proba.limit}
&\sum_{l=1}^L R_l \PR{\R_l|\nack}<\R_L,
\end{align}
where \eqref{proba.limit} is obtained by taking the limits of both sides of~\eqref{proba.ineq} for~$\SNRav\rightarrow\infty$.

If, for all $l<L$, $\PR{\R_l|\nack}>0$, then~\eqref{proba.limit} is satisfied. This is another version of the conditions required in \propref{Prop:amc.ge.harq}. It can be verbalized as follows: knowing that an error occurred (\ie $\nack$ is observed), the a posteriori probability of choosing a rate $R_l$, $l=1,\ld, L-1$ does not go to zero. This means that, for the high-\gls{snr} regime, the error event $\nack$ provides information about the realization of the \gls{snr} which led to the error.

The most obvious interpretation of this relationship is obtained for $\tilde{a}=\infty$ with the decision regions $\mcD^\harq_l=\mcD^\amc_l$, where $\gamma_1=0$ and $\gamma_l=\SNR_{\tnr{th},l}$ for  $l>1$. Then, the only source of error is the event $\SNRrv\in (0, \SNR_{\tnr{th},1} )$. Since this event occurs only when the rate $\R_1$ is selected, we can write $\PR{\R_1|\nack}=1$. That is, even if the a priori probability $p_1=\PR{\R=\R_1}$ can be arbitrarily small, the a posteriori probability can be very large.

%!TEX root =  ../HARQ-AMC.tex
%%%%%%%%%%%%%%%%%%%%%%%%%%%%%%%%%%%%%%%%%%%%%%%%%%
\section{Improving the Throughput via AMC-HARQ Interaction}\label{Sec:Boosting.HARQ}

With the lessons from \secref{Sec:Block.fading}, we explore in this section two simple solutions to limit the throughput penalty incurred by the adoption of \gls{harq} on top of the \gls{amc} in fast-fading channels. We do it in two steps. First, we will try to remove the source of the throughput penalty by making \gls{harq} aware of the \gls{csi}; next, we will explore the possibility of increasing the throughput with a more advanced coding.

%%%%%%%%%%%%%%%%%%%%%%%%%%%%
\subsection{AMC-Aware HARQ: Packet dropping}\label{Sec:Aware.harq}
The main point from the previous analysis is that, while the receiver observes the \gls{csi}, the adaptation has a very limited scope: the parameters of \gls{harq} are fixed in the first round, and the subsequent ones ($k=2,\ld, K$) are not affected by the varying \gls{csi}.

In \secref{Sub:discussion}, we have also identified situations where \gls{harq} used on top of the \gls{amc} is counterproductive: suppose we observe in the $k$-th round the \gls{mcs} index, $\hat{l}_k$. Using the \gls{amc} the reward can be close to $\R_{\hat{l}_k}$, see \eqref{eta.l.amc}. On the other hand, the reward related to \gls{harq} depends on the first round \gls{mcs} index, $\hat{l}_1$, and may be close to $\R_{\hat{l}_1}$. Therefore, if $\R_{\hat{l}_k}$ is greater than $\R_{\hat{l}_1}$, the  \gls{amc} should be used; this modifies the protocol for rounds $k>1$ in the following manner
\begin{align}
\label{cont.harq}
\text{if }\hat{l}_k \leq \hat{l}_1 &\implies\quad \text{continue HARQ}\\
\label{stop.harq}
\text{if }\hat{l}_k > \hat{l}_1 &\implies\quad \text{restart HARQ (drop the packet)}.
\end{align}

In this way, we adapt the number of \gls{harq} transmission rounds to the observed \gls{mcs} indices. Quite clearly, only the performance of \gls{harq} in fast-fading channels will be affected because, under slow-fading model, we observe $\hat{l}_1=\hat{l}_2=\ld=\hat{l}_K$, and subsequently \eqref{cont.harq} applies in all rounds.

Using \eqref{stop.harq}, the \gls{harq} the number of rounds of one cycle may depend on the rate used in the cycle which follows. Thus, to calculate the throughput of the \gls{pdharq} we have just introduced,  we cannot directly use the previous formulas based on the expected reward and the expected duration in \emph{one} \gls{harq} cycle. 

However, this is not an issue since we are merely interested in the throughput with fixed transmission strategy parameters, and we will also use the predefined decision regions of \gls{amc}, $\mcD^\amc_l$. Therefore, to obtain the results we show in \figref{Fig:MultiPacket}, we used a Monte Carlo integration. We can appreciate that the packet dropping proposed in \eqref{stop.harq} practically eliminates the throughput loss introduced by \gls{harq}. This is an indirect confirmation that the main source of the throughput gap was correctly identified in \secref{Sub:discussion}.

%%%%%%%%%%%%%%%%%%%%%%%%%%%%%%%%%
\begin{figure}[ht]
\psfrag{xlabel}[ct][ct][\siz]{$\SNRav$~[dB] }
\psfrag{ylabel}[ct][ct][\siz]{$\gamma_l$ [dB]}
\psfrag{G=0.5}[cl][cl][\siz]{$\tilde{a}=0.5$}
\psfrag{G=1}[cl][cl][\siz]{$\tilde{a}=1$}
\psfrag{G=4}[cl][cl][\siz]{$\tilde{a}=4$}
\psfrag{AMC}[cl][cl][\siz]{\gls{amc}}
\psfrag{HARQ-packet-drop}[cl][cl][\siz]{\gls{pdharq}}
\psfrag{HARQ-VL}[cl][cl][\siz]{\gls{vlharq}}
\begin{center}
\scalebox{\sizf}{\includegraphics[width=\linewidth]{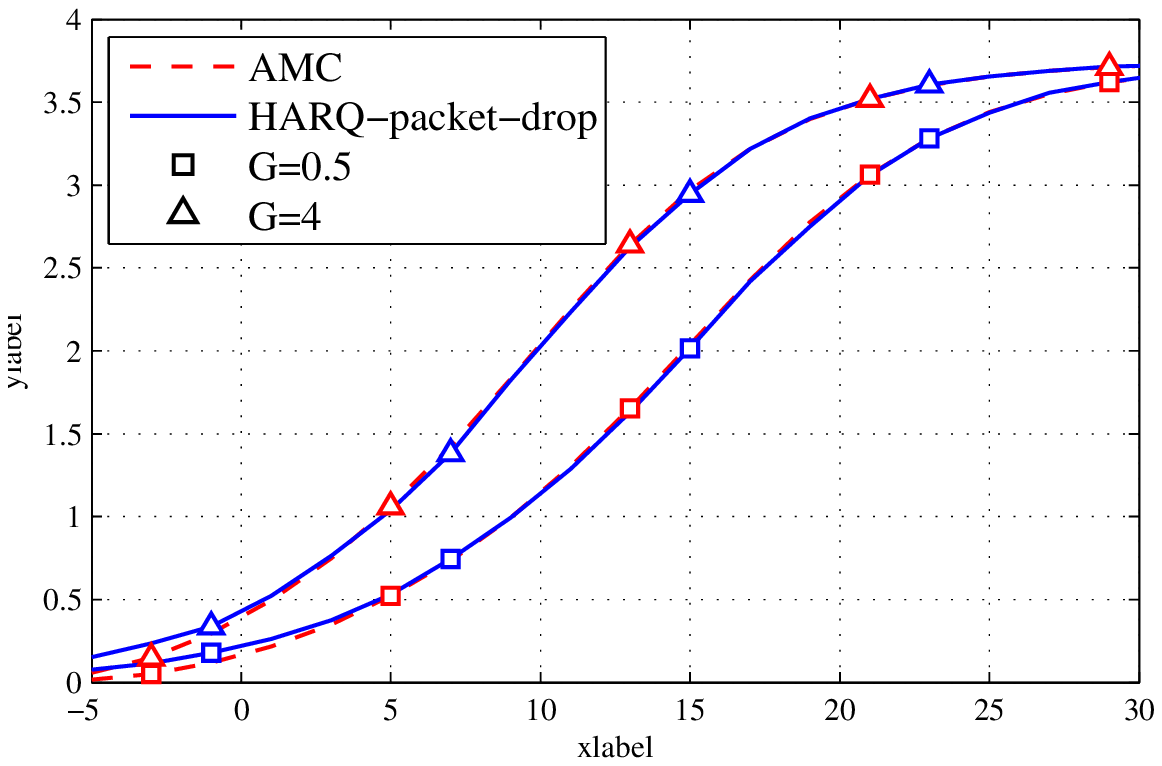}}\\
\scalebox{\siz}{a)}\\
\scalebox{\sizf}{\includegraphics[width=\linewidth]{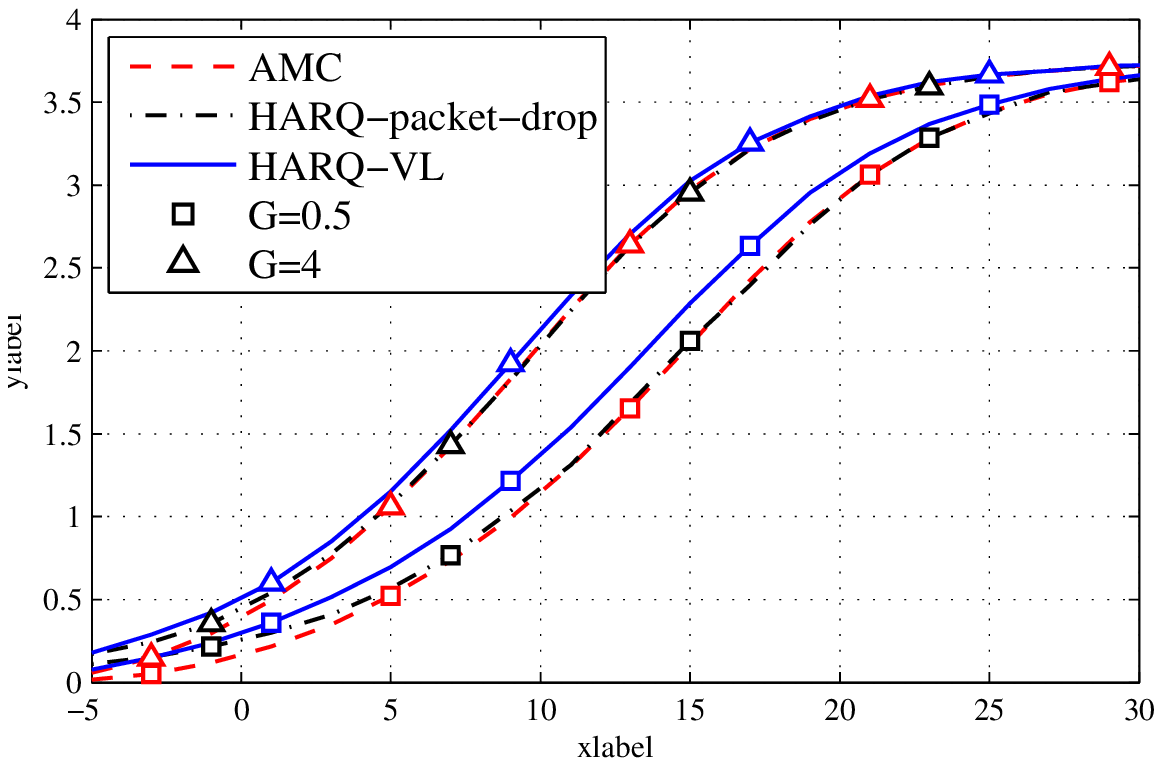}}\\
\scalebox{\siz}{b)}
\caption{The throughput obtained using packet-dropping \gls{pdharq} for a)~\gls{rrharq}, and  b)~\gls{irharq}, and using variable-length \gls{vlharq}, which can be implemented only for \gls{irharq}.}\label{Fig:MultiPacket}
\end{center}
\end{figure}
%%%%%%%%%%%%%%%%%%%%%%%%%%%%%%%%%

%%%%%%%%%%%%%%%%%%%%%%%%%%%%%%%%%
\subsection{Adaptive HARQ: Variable-Length Coding}\label{Sec:vlharq}
Here, we would like to  improve the throughput beyond what \gls{amc} alone can offer, especially in the high-\gls{snr} regime. We follow \cite{Uhlemann03,Visotsky03,Gopalakrishnan08, Kim11,Szczecinski13}, which proposed to use variable-length coding  for \gls{harq}; the main idea of \gls{vlharq} is to reduce the number of symbols involved in the transmission of a packet. This can be done by decreasing the length of the codewords in rounds $k=2,\ld, \kmax$.

In order to comply with the constant-length channel block assumptions we used previously, modifications are needed regarding  the way the packets are encoded and transmitted in a channel block. 

First, we redefine the operation of the \gls{amc}. Previously, the rate $\R_l$ was interpreted as a transmitting of one packet containing $\R_l\Ns$ bits per block.  We assume henceforth that the packet contents is fixed to $\Nb=R_1\Ns$ bits so selecting the rate $\R_l$ means that the transmission requires $\Nsl{l}=\Nb/\R_l=\ell_l\Ns$ symbols, where $\ell_l=R_1/R_l$ is the normalized length of the codeword.\footnote{Equivalently, $\ell_l$ is the fraction of $\Ns$ required to transmit the packet with the rate $\R_l$.} If we also use $\R_l=l \R_1$, $l=1,\ld, L$ (as we already did in the numerical examples), then the \gls{amc} transmission with the rate $\R_l$ means that $l$ packets are transmitted in one channel block, as done also in \cite{Liu04,Le06,Szczecinski13}. 

Since the relationship between the rate $\R_l$ and the  length $\ell_l$ is bijective, it is convenient to reformulate the discussion in terms of packet length. The first transmission round of \gls{harq} is done using codewords $\bx_1$ with the predefined length taken from the set  $\mcL=\set{1,1/2,\ld, 1/L}$; we are thus compatible with the \gls{amc}. On the other hand, in the subsequent \gls{harq} rounds we are allowed to use the subcodewords $\bx_k$ with arbitrary length taken from $\mcL$. Since the idea is to use shorter subcodewords in order to decrease the channel use, we may also expand the decision space and use auxiliary set of codewords lengths $\mcL'={\ell_{L+1}, \ld, \ell_{L'}}$, where $\ell_{k-1}<\ell_{k}, k=L+1,\ld, L'$. 

Since it is allowed to transmit in the same channel block the ``fresh'' (not transmitted) packets as well as the redundancy codewords for the \gls{nack}ed packets, the operation of \gls{harq} needs to be modified. 

%% main problem solved
To do so, we assume that all packets, which are candidates for the transmission are gathered in the \gls{harq} buffer and we index them with $h\in\set{1,\ld, H}$, where $H\leq \kmax\cd L $ is the maximum number of packets we need to consider. Let $\set{\ell_h}_{h=1}^H$ denote the length we will assign to each packets in the \gls{harq} buffer, where setting $\ell_h=0$ means, de facto, that the packet is not ``scheduled" for transmission.

We will maximize the instantaneous reward for a given \gls{snr}, $\SNR$. As all the packets carry $\Nb$ bits, we solve the following problem:
\begin{align}\label{Opt.MultiPacket.Throughput}
\set{\hat{\ell}_h}&=\argmax_{\set{\ell_h}} \sum_{h=1}^H (1-f(h)),
\quad\text{s.t.}\quad \sum_{h=1}^H \ell_h\leq1
\end{align}
where $f(h)$ is the conditional probability of error for the $h$-th packet, calculated as
\begin{align}\label{f.h.def}
f(h)=\PR{\err_{k(h)+1}|\err_{k(h)}}
\end{align}
and $k(h)$ is the \gls{harq} counter of the $h$-th packet (\ie the number of time the packet was scheduled for transmission).

%%% technicality of error calculation
To  calculate \eqref{f.h.def} we first express  the probability of decoding error after $k(h)$ transmissions with the codewords' lengths $\ell_{1}(h),\ld,\ell_{k(h)}(h)\in\set{\mcL,\mcL'}$ using \eqref{MD.2.1D}--\eqref{aggreg.SNR}
\begin{align}\label{PER:multipacket}
\PR{\err_k|\ell_{1}(h),\ld,\ell_{k(h)}(h) }=\PER_{l(h)}(\SNR^\Sigma(h)),
\end{align}
where, $l(h)$ is the index of the first transmission (\ie $\ell_{l(h)}=\ell_{1}(h)$), and for \gls{irharq} the $\SNR^\Sigma(h)$ is given by 
\begin{align}\label{aggreg.SNR.Multi}
\SNR^\Sigma(h)=I^{-1}\Big( \frac{1}{\ell_{1}(h)}\sum_{t=1}^{k(h)} \ell_{t}(h) I(\SNR_t(h))\Big),
\end{align}
where $\SNR_{t}(h)$ is the \gls{snr} experienced by the $h$-th packet during the $t$-th round and \eqref{aggreg.SNR.Multi} expresses the fact that the \gls{mi} accumulation depends on the subcodeword lengths. Note that, as expected, \eqref{PER:multipacket} is equivalent to \eqref{MD.2.1D} if $\ell_{1}(h)=\ld=\ell_{k(h)}(h)$.

Thus, using \eqref{NACK.ERR} we calculate \eqref{f.h.def} as
\begin{align}
f(h)=\frac{\PER_{l(h)}(\SNR'(h))}{\PER_{l(h)}(\SNR^\Sigma(h))},
\end{align}
where
\begin{align}\label{SNR.prime}
\SNR'(h) &= I^{-1}\Big(I(\SNR^\Sigma(h))+\frac{ \ell_h}{\ell_1(h)}\SNR\Big)
\end{align}
is the aggregate \gls{snr} that will be experienced by the $h$-th packet provided it is scheduled for transmission.

After erroneous decoding  of the $h$-th packet, signaled by a \gls{nack} message, $k(h)$ is incremented if $k(h)<\kmax-1$ (otherwise the packet is discarded from the \gls{harq} buffer), and $\SNR^\Sigma(h)$ is updated using \eqref{aggreg.SNR.Multi} 
\begin{align}\label{SNR.sigma.next}
 \SNR^\Sigma(h)&\leftarrow\SNR'(h).
\end{align}

We thus see that it is not necessary to keep track of all lengths used to transmit  the packet $h$: observing the current \gls{snr}, $\SNR$, and knowing the triplet $\big(k(h), \ell_{1}(h), \SNR^{\Sigma}(h)\big)$, for all the packets $h=1,\ld, H$, we have enough information to calculate the probability of error for each scheduled packet.

We note that, if all the packets in the buffer are fresh, the optimization problem \eqref{Opt.MultiPacket.Throughput} is equivalent to finding the throughput-maximizing rates as done for \gls{amc} when the decision regions are defined by~\eqref{mcD.l}. On the other hand, in presence of packets with \gls{harq} counters $k(h)>0$, the notion of decision regions is less clear; since they become dependent on the parameters of all the packets in the \gls{harq} buffer via \eqref{Opt.MultiPacket.Throughput}, the  simple one-to-one relationship between the packet and the transmission rate does not exist anymore.

To make a comparison with the \gls{amc} fair, we assume that in order to have a decodable first transmission (\ie $f(h)<1$) we must use the length $\ell\in\mcL$. Again, this is motivated by the practical considerations: for example, using a finite-size constellation, such as 16-\gls{qam}, any transmission with rate $\R>4$ (or $\ell<R_1/4$) fails.  Therefore, using the instantaneous throughput criterion, the first transmission length must be taken from $\mcL$ so we cannot exploit $\mcL'$ to improve the throughput of the \gls{amc}. 

%% comment on the suboptimality
We terminate noting  that the memoryless scheduling decisions obtained solving \eqref{Opt.MultiPacket.Throughput} are suboptimal in terms of long-term throughput. In general, finding the optimal set $\set{\hat{\ell}_h}_{h=1}^H$ is difficult mainly due to the dependence of the throughput on the future decisions and would have to be solved using a framework of \gls{mdp}; a problem is compounded by the dimensionality of the observation space.

The numerical results we show in \figref{Fig:MultiPacket} are obtained by Monte Carlo simulations, where the set of lengths $\mcL=\set{1,1/2,1/3,1/4,1/5}$ is the same as in all previous examples, and the auxiliary set is given by $\mcL' =\set{1/8, 1/12, 1/16}$. The discrete optimization problem defined in \eqref{Opt.MultiPacket.Throughput} was solved  by exhaustive search over the entire solution space; this is possible here due to a relatively small dimensionality of the sets $\mcL$ and $\mcL'$. 

The improvement of the throughput of \gls{vlharq} comparing to \gls{amc}-\gls{harq} is clear for high values of $\SNRav$ despite of the sub-optimality of the objective function \eqref{Opt.MultiPacket.Throughput}. We conducted more experiments including smaller length $\ell$ in $\mcL'$, but changes in the throughput we observed were not significant. On the other hand, it is was indeed necessary to use $\mcL'$ in order to improve the throughput in high \gls{snr}. This is because, under such operating conditions, the first transmission is very likely to use the shortest codewords from $\mcL$ and thus, to decrease the average transmission time, we need to use shorter codewords from $\mcL'$.

\section{Conclusions}\label{Sec:Conclusions}
In this work, we provided a general communication-theoretic point of view on the benefits of combining \gls{amc} and \gls{harq} in point-to-point transmissions. Using the throughput as a comparison criterion, the main conclusions are the following: i)~in slow-fading channels, the throughput increases thanks to \gls{harq} but the improvements are very moderate, especially when \gls{amc} ignores the presences of \gls{harq}, ii)~in fast-fading channels, \gls{harq} is beneficial only in low \gls{snr}s and is counterproductive in high \gls{snr}s, irrespectively of the adaptation efforts of \gls{amc}. Then, since \gls{harq} provides very small (if any) throughput gains, the error-free operation may be taken in charge by the upper layer (\gls{llc}). 

%Furthermore, using \gls{amc} rate-decision regions does not change significantly the throughput and has advantage that the first transmission is always ``decodable'', that is, in the first \gls{harq} round we do not rely only on the future channel realizations to decode the message. Again, this would not be proscribed but might be a ``risky'' approach because the mathematical modeling of the channel state must be backed by reality, i.e., in the case the \gls{snr} does not follow the slow-fading model; the transmission results become unpredictable.

Following the identified sources of deficiencies of \gls{harq} used on top of \gls{amc}, we also proposed a simple modification to the \gls{harq} protocol, which terminates the \gls{harq} cycle on the basis of the observed \gls{mcs}. This modification can be implemented without any additional signaling or change at \gls{phy}, and removes the observed throughput penalty in fast fading channels. We pursued the idea of a closer interaction between \gls{phy} and \gls{harq}, and evaluated the possibility of using packet dropping (\gls{pdharq}) which provide a simply remedy to the throughput penalty. At the cost of more complicated implementation, the variable-length coding (\gls{vlharq}) leads to a moderate---yet notable---improvement of the throughput in the region of high \gls{snr}.

The conclusions and observations we make hold for delay-insensitive applications, and do not take into account the overhead necessary to retransmit the packet at the \gls{llc}. Thus,  a more realistic evaluation should differentiate between the cost of \gls{phy} and \gls{llc} transmissions. Other important challenges would be to evaluate the performance of \gls{harq} and \gls{amc} in the interference-limited scenario or for a delay-sensitive traffic.

%\begin{appendices}
%\input{./include/Appendix.tex}
%\end{appendices}

\balance

%\bibliography{IEEEabrv,references_all}
%\bibliographystyle{IEEEtran}

\bibliographystyle{IEEEtran}

\end{document}